\documentclass[10pt]{article}   % Specifies the document style.
\usepackage{latexsym}     % Author's package choice
\usepackage{enumerate}
\usepackage{ams}  
\usepackage{amsthm}     
\usepackage{amssymb}     
\usepackage{amsmath}     
\usepackage{proof}     
\usepackage{graphics}     
\usepackage{color}     
\usepackage[all]{xy}    
\newtheorem{theorem}{Theorem}[section]
\newtheorem{lemma}[theorem]{Lemma}

\newtheorem{proposition}[theorem]{Proposition}

\theoremstyle{definition}

\newcommand{\bbP}{\ensuremath{\mathbb{P}}}  % Mathblackboard P (powerset)

\newcommand{\bP}{\ensuremath{\mathbf{P}}}   % Mathbold P (set of unary prediates)
\newcommand{\cL}{\ensuremath{\mathcal{L}}}  % Caligraphic L (general langauge)
\newcommand{\cH}{\ensuremath{\mathcal{H}}}  % Caligraphic H (hamiltonian-syllogistic)
\newcommand{\cS}{\ensuremath{\mathcal{S}}}  % Caligraphic S (classical syllogistic)
\newcommand{\fA}{\ensuremath{\mathfrak{A}}} % Fraktur A (structure)
\newcommand{\fB}{\ensuremath{\mathfrak{B}}} % Fraktur B (structure)
\newcommand{\fC}{\ensuremath{\mathfrak{C}}} % Fraktur C (structure)
\newcommand{\fD}{\ensuremath{\mathfrak{D}}} % Fraktur D (structure)

\newcommand{\sX}{{\sf X}}                   % Sans-serif X for rule systems
\newcommand{\sH}{{\sf H}}                   % Sans-serif H for rule systems

\newcommand{\set}[1]{\{#1\}}
\renewcommand{\phi}{\varphi}

\newcommand{\PTIME}{\textsc{PTime}}
\newcommand{\NPTIME}{\textsc{NPTime}}

\title{The Hamiltonian Syllogistic}
\author{ Ian Pratt-Hartmann} \date{}
\begin{document}
\maketitle
\begin{abstract}
\noindent
This paper undertakes a re-examination of Sir William Hamilton's
doctrine of the {\em quantification of the predicate}. Hamilton's
doctrine comprises two theses. First, the predicates of traditional
syllogistic sentence-forms contain implicit existential quantifiers,
so that, for example, {\em All $p$ are $q$} is to be understood as
{\em All $p$ are {\em some} $q$}. Second, these implicit quantifiers can
be meaningfully dualized to yield novel sentence-forms, such as, for
example, {\em All $p$ are {\em all} $q$}. Hamilton attempted to
provide a deductive system for his language, along the lines of the
classical syllogisms. We show, using techniques unavailable to
Hamilton, that such a system does exist, though with qualifications
that distinguish it from its classical counterpart.
\end{abstract}
\section{Introduction}
\label{sec:introduction}
By the {\em classical syllogistic}, we understand the set of English
sentences of the forms
\begin{equation}
\begin{array}{lll}
\mbox{Every $p$ is a $q$ \hspace{2.5cm}} 
&  & 
\mbox{Some $p$ is a $q$ \hspace{2.8cm}}\\
\mbox{No $p$ is a $q$} & & 
\mbox{Some $p$ is not a $q$},
\end{array}
\label{eq:EnglishS}
\end{equation}
where $p$ and $q$ are common (count) nouns.  These sentence-forms are
evidently logically equivalent to the following more cumbersome forms:
\begin{equation}
\begin{array}{lll}
\mbox{Every $p$ is identical to some $q$} 
\hspace{0cm} & &
\mbox{Some $p$ is identical to some $q$}\\
\mbox{No $p$ is identical to any $q$} & & 
\mbox{Some $p$ is not identical to any $q$}.
\end{array}
\label{eq:EnglishH1}
\end{equation}
By the {\em Hamiltonian syllogistic}, we understand the set of
sentences of the forms~\eqref{eq:EnglishH1}, together with the set of
sentences of the forms
\begin{equation}
\begin{array}{lll}
\mbox{Every $p$ is identical to every $q$} 
\hspace{0cm} & &
\mbox{Some $p$ is identical to every $q$}\\
\mbox{No $p$ is identical to every $q$} & &
\mbox{Some $p$ is not identical to every $q$},
\end{array}
\label{eq:EnglishH2}
\end{equation}
obtained from~\eqref{eq:EnglishH1} by dualizing the second quantifier.
Taking account of the equivalence of the forms~\eqref{eq:EnglishS}
and~\eqref{eq:EnglishH1}, we informally regard the classical
syllogistic as a subset of the Hamiltonian syllogistic. The
sentence-forms~\eqref{eq:EnglishH2} have no idiomatic English
equivalents. We take their respective first-order translations to be
\begin{equation}
\begin{array}{lll}
\forall x(p(x) \rightarrow \forall y(q(y) \rightarrow x = y)) \hspace{0.5cm} & & 
\exists x(p(x) \wedge \forall y(q(y) \rightarrow x = y)) \hspace{0.5cm}\\
\forall x(p(x) \rightarrow \neg \forall y(q(y) \rightarrow x = y)) & &  
\exists x(p(x) \wedge \neg \forall y(q(y) \rightarrow x = y)).
\end{array}
\label{eq:FOLH2}
\end{equation}
Thus, for example, {\em Every $p$ is identical to every $q$} is true
just in case there are no \mbox{$p$\hspace{0.5mm}s}, or no $q$s, or
there is exactly one $p$ and one $q$, and they are identical. Observe
that determiners in subjects are taken to outscope those in
predicates.  Observe also that verb negation is taken to outscope the
following predicate determiner.  Thus, for example, {\em Some $p$ is
  not identical to every $q$} is true just in case some $p$ is
distinct from some $q$.

By a {\em classical syllogism} we understand any of the valid
two-premise argument patterns for sentences of the
forms~\eqref{eq:EnglishS}, for example:
\begin{center}
\begin{minipage}{5cm}
\begin{tabbing}
Every $p$ is a $q$\\
\underline{Every $q$ is an $r$}\\
Every $p$ is an $r$
\end{tabbing}
\end{minipage}
\hspace{4cm}
\begin{minipage}{5cm}
\begin{tabbing}
Every $q$ is an $r$\\
\underline{Some $p$ is a $q$}\\
Some $p$ is an $r$.
\end{tabbing}
\end{minipage}
\end{center}
It is known that the classical syllogisms---with one or two minor
additions---constitute a sound and complete proof system for the
classical syllogistic. Does there exist a comparable system
of rules for the Hamiltonian syllogistic?

The Hamiltonian syllogistic is so called because of its
more-than-passing resemblance to Sir William Hamilton's doctrine of
the {\em quantification of the predicate}. (That is: Sir
William Hamilton, Bart., the Scottish philosopher, not Sir William
Hamilton, Kt., the Irish mathematician who discovered quaternions.)
According to that doctrine, the predicates of the traditional
syllogistic forms
\begin{equation*}
\begin{array}{lll}
\mbox{All $p$ are $q$} \hspace{3cm} 
& & 
\mbox{Some $p$ are $q$} \hspace{3cm} \\
\mbox{No $p$ are $q$} & & 
\mbox{Some $p$ are not $q$},
\end{array}
\end{equation*}
contain a suppressed existential quantifier (present, as Hamilton put
it, {\em in thought}), which can meaningfully be dualized to yield the forms
\begin{equation*}
\begin{array}{lll}
\mbox{All $p$ are all $q$} \hspace{1.5cm} 
& \hspace{1cm} & 
\mbox{Some $p$ are all $q$}  \hspace{2cm} \\
\mbox{No $p$ are all $q$} & & 
\mbox{Some $p$ are not all $q$}.
\end{array}
\end{equation*}
These latter sentences are of course, grammatically marginal, and we
are owed an account of their purported meaning. Unfortunately,
Hamilton's presentation is hopelessly obscure in this regard: the
closest we get to a formal account are the collections of diagrams
in~\cite{qpr:hamilton53}, pp.~682--683 and~\cite{qpr:hamilton60},
p.~277.  However, it is certain, from the accompanying text, that
Hamilton took {\em All $p$ are all $q$} to assert that $p$ and $q$ are
coextensive---different from the meaning of the formula $\forall
x(p(x) \rightarrow \forall y(q(y) \rightarrow x = y))$, and logically
uninteresting.

Hamilton originally expounded his theory in an 1846 edition of the
works of Sir Thomas Reid, in the form of a Prospectus of ``An Essay
towards a New Analytic of Logical Forms'', reproduced, with some
variations, in the two sources mentioned above. The essay itself was
never written. Nevertheless, Hamilton's theory generated a heated
debate with Augustus De Morgan, and was the subject of a critical
commentary by John Stuart Mill (\cite{mill79}, Ch.~XXII).  The present
paper attempts neither to resurrect that debate, nor to adjudicate its
outcome. Unlike Hamilton's quantified predicates, the sentence-forms
in~\eqref{eq:EnglishH1} and~\eqref{eq:EnglishH2} are clearly
grammatical, and the question of the existence of sound and complete
proof procedures for this language consequently well-formed. Of
course, to be well-formed is one thing; to be well-motivated,
another. We mention just one striking historical fact by way of
justification. Notwithstanding their dubious grammatical status,
copula sentences with quantified predicates make regular appearances
in discussions of the syllogism, beginning with Aristotle himself
(see~\cite{qpr:aristotlePA}, A27, 43a12--43b22). Indeed, Hamilton's
own extensive survey of this literature can be found
in~\cite{qpr:hamilton60}, pp.~298-317. Why, if we happily judge {\em
  No pacifist admires every quaker} to be grammatical, are we much
less comfortable with {\em No pacifist {\em is} every quaker}?  What,
if anything, is this non-grammaticality judgement preventing us from
expressing? And if there is something, what would be the logical
consequences of expressing it anyway?  Thus, our investigation of the
Hamiltonian syllogistic addresses a venerable, if currently quiescent,
issue.

On the basis of the equivalence of the forms~\eqref{eq:EnglishS}
and~\eqref{eq:EnglishH1}, we take the best candidates for syllogistic
forms with universally quantified predicates to
be~\eqref{eq:EnglishH2}, interpreted as~\eqref{eq:FOLH2}.  We show in
the sequel that, under this interpretation, no finite set of
syllogistic rules can be sound and complete for the Hamiltonian
syllogistic, a fact which distinguishes it from its classical
subset. However, we do provide a finite set of such rules which is
sound and {\em refutation}-complete---i.e.~becomes complete if the
rule of {\em reductio ad absurdum} is permitted as a final step. We go
on to consider the effect of adding noun-level negation to the
Hamiltonian syllogistic, yielding such sentence-forms as {\em No
  non-$p$ is identical to every non-$q$}. We show that, unless
\PTIME=\NPTIME, no finite set of syllogistic rules can be sound and
refutation-complete for this extended language. However, we do provide
a finite set of such rules which are sound and complete if the rule of
{\em reductio ad absurdum} may be used without restriction. Such
sensitivity to noun-level negation again distinguishes the Hamiltonian
syllogistic from its classical subset.

By replacing the words {\em is identical to} in~\eqref{eq:EnglishH1}
and~\eqref{eq:EnglishH2} by a suitably inflected transitive verb $v$
({\em admire}, {\em despise} \ldots), we obtain the forms
\begin{align*}
& \begin{array}{lll}
\mbox{Every $p$ $v$s some $q$} 
& \hspace{2cm} & 
\mbox{Some $p$ $v$s  some $q$}\\
\mbox{No $p$ $v$s  any $q$} & & 
\mbox{Some $p$ does not $v$ any $q$}\\
\mbox{Every $p$ $v$s  every $q$} 
& & 
\mbox{Some $p$ $v$s  every $q$}\\
\mbox{No $p$ $v$s  every $q$} & & 
\mbox{Some $p$ does not $v$ every $q$}.
\end{array}
\end{align*}
This language was analysed in Pratt-Hartmann and
Moss~\cite{qpr:p-h+m09}, where it was called the {\em relational
  syllogistic}.  It was shown there that no finite set of syllogistic
rules in the relational syllogistic is sound and complete, though
there is a finite set of rules that is sound and
refutation-complete. It was shown in the same paper that, when the
relational syllogistic is extended with noun-level negation, there is
no finite set of syllogistic rules that is sound and complete for the
resulting language, even when the rule of {\em reductio ad absurdum}
may be used without restriction. Thus, the Hamiltonian syllogistic
differs in its proof-theoretic properties from the relational
syllogistic as well.

\section{Syntax and semantics}
\label{sec:synsem}
In this section, we define five formal languages: ({\em i}) $\cS$, a
formalization of the classical syllogistic, ({\em ii}) $\cS^\dagger$,
its extension with noun-level negation, ({\em iii}) $\cH$, a
formalization of the Hamiltonian syllogistic, ({\em iv})
$\cH^\dagger$, its extension with noun-level negation, and ({\em v})
$\cH^{*\dagger}$, an extension of $\cH^\dagger$ motivated chiefly by
the formalism used below.

Fix a countably infinite set $\bP$.  We refer to any element of $\bP$
as an {\em atom}.  A {\em literal} is an expression of either of the
forms $p$ or $\bar{p}$, where $p$ is an atom.  
A literal which is an atom is called {\em positive},
otherwise, {\em negative}. If $\ell = \bar{p}$ is a negative literal,
then we denote by $\bar{\ell}$ the positive literal $p$. A {\em
  structure} is a pair $\fA = \langle A, \{ p^\fA \}_{p \in \bP}
\rangle$, where $A$ is a non-empty set, and $p^\fA \subseteq A$, for
every $p \in \bP$. The set $A$ is called the {\em domain} of $\fA$. We
extend the map $p \mapsto p^\fA$ to negative literals by setting, for
any atom $p$,
\begin{equation*}
\bar{p}^\fA = A \setminus p^\fA.
\end{equation*}
Intuitively, we may think of the elements of $\bP$ as common
count-nouns, such as {\em pacifist}, {\em quaker}, {\em republican},
etc., and we think of $p^\fA$ as the set of things falling under the
noun $p$ according to the structure $\fA$.  Thus, we may read $\bar{p}$ as
either {\em non-$p$} or {\em not a $p$}, depending on grammatical
context.

An $\cS$-{\em formula} is any expression of the forms
\begin{equation}
\begin{array}{lllll}
  \forall(p,q) & \hspace{1cm} &   \forall(p,\bar{q}) & \hspace{1cm} & 
  \forall(\bar{p},\bar{q})\\
\exists(p,q) & \hspace{1cm} & \exists(p,\bar{q}) & \hspace{1cm} & \exists(\bar{p},q),
\end{array}
\label{eq:syntaxS}
\end{equation}
where $p$ and $q$ are atoms; and an $\cS^\dagger$-{\em formula} is any
expression of the forms
\begin{equation}
\begin{array}{lll}
\forall(\ell,m) & \hspace{2cm} & \exists(\ell,m),
\end{array}
\label{eq:syntaxSdag}
\end{equation}
where $\ell$ and $m$ are literals. Thus, every $\cS$-formula is an
$\cS^\dagger$-formula. If $\fA$ is a structure, we write $\fA \models
\forall(\ell,m)$ if $\ell^\fA \subseteq m^\fA$, and $\fA \models
\exists(\ell,m)$ if $\ell^\fA \cap m^\fA \neq \emptyset$.  We think of
$\fA \models \phi$ as asserting that $\phi$ is true in the structure
$\fA$.  Thus, we may read $\forall(\ell,m)$ as {\em Every $\ell$ is an
  $m$} and $\exists(\ell,m)$ as {\em Some $\ell$ is an $m$}.  Under
these semantics, the formulas $\exists(\ell,m)$ and $\exists(m,\ell)$
are true in exactly the same structures; and similarly for the
formulas $\forall(\ell,m)$ and $\forall(\bar{m},\bar{\ell})$. In the
sequel, we identify these pairs of formulas, silently converting one
into the other where needed.  Taking account of these identifications,
$\cS$ contains four different formulas (not six), which may be glossed
by the sentence-forms~\eqref{eq:EnglishS} or, equivalently,
\eqref{eq:EnglishH1}. Likewise, $\cS^\dagger$ contains six different
formulas (not eight), where $\forall(\bar{p}, q)$ is to be glossed as
{\em Every non-$p$ is a $q$}, and $\exists(\bar{p}, \bar{q})$ as {\em
  Some non-$p$ is not a $q$}.

Turning now to the languages $\cH$, $\cH^\dagger$ and $\cH^{*\dagger}$, a
{\em c-term} is either a literal or any expression of the forms
$\forall p$ or $\overline{\forall p}$, where $p$ is an atom; and an
{\em e-term} is either a literal or any expression of the forms
$\forall \ell$ or $\overline{\forall \ell}$, where $\ell$ is a
literal. Thus, all literals are c-terms, and all c-terms are
e-terms. If $e$ is an e-term of the form $\overline{\forall \ell}$, we
take $\bar{e}$ to be the corresponding e-term $\forall \ell$. It
follows that $\bar{e}$ is a c-term if and only if $e$ is, and
$\bar{e}$ is a literal if and only if $e$ is; moreover, for any e-term
$e$, $\bar{\bar{e}} = e$.  If $\fA$ is a structure, we extend the map
$\ell \mapsto \ell^\fA$ to non-literal e-terms by setting, for any
literal $\ell$,
\begin{align*}
(\forall \ell)^\fA & = \{a \in A \mid \mbox{$a = b$ for all  $b \in \ell^\fA$}\}\\
(\overline{\forall \ell})^\fA & = \{a \in A \mid \mbox{$a \neq b$ for some  $b \in \ell^\fA$}\},
\end{align*}
Thus, we may read $\forall \ell$ as {\em thing that is identical to
  every $\ell$}, and $\overline{\forall \ell}$ as {\em thing that is
  not identical to every $\ell$} (that is: {\em thing that is distinct
  from some $\ell$}). Because terms of the form $\forall \ell$ can be
confusing to parse in certain contexts, we sometimes enclose them in
parentheses, thus: $(\forall \ell)$.

An $\cH$-{\em formula} is any expression of the forms
\begin{equation*}
\begin{array}{lll}
\forall(p, c) & \hspace{2cm} & \forall(c, \bar{p})\\
\exists(p, c) & & \exists(c,p),
\end{array}
\end{equation*}
where $p$ is an atom and $c$ is a c-term; an $\cH^\dagger$-{\em
  formula} is any expression of the forms
\begin{equation*}
\begin{array}{lll}
\forall(\ell, e) & \hspace{2cm} & \exists(\ell, e),
\end{array}
\end{equation*}
where $\ell$ is a literal and $e$ is an e-term; and an
$\cH^{*\dagger}$-{\em formula} is any expression of the forms
\begin{equation*}
\begin{array}{lll}
\forall(e,f) & \hspace{2cm} & \exists(e,f),
\end{array}
\end{equation*}
where $e$ and $f$ are e-terms.  Thus, every $\cS^{\dagger}$-formula is
an $\cH$ formula, every $\cH$-formula is an $\cH^{\dagger}$ formula,
and every $\cH^{\dagger}$-formula is an $\cH^{*\dagger}$ formula.  We
define $\fA \models \phi$ for $\cH^{*\dagger}$-formulas in the same
way as for $\cS^{\dagger}$-formulas, again silently converting
$\exists(e,f)$ to $\exists(f,e)$, and $\forall(e,f)$ to
$\forall(\bar{f},\bar{e})$, as needed.  Taking account of these
conversions, the eight forms of $\cH$-formula may be glossed using the
sentence-forms~\eqref{eq:EnglishH1} and~\eqref{eq:EnglishH2}; and
$\cH^\dagger$-formulas may be similarly glossed, but using negated
noun-phrases such as {\em non-pacifist}, {\em non-quaker}, etc.~in the
obvious way.  Formulas of the language $\cH^{*\dagger}$, by contrast,
require more elaborate English translations: for example,
$\forall(\forall \bar{p},\overline{\forall q})$ may be glossed as:
\begin{quote}
  Everything that is identical to every non-$p$ is distinct from some
  $q$. 
\end{quote}
The primary motivation for considering the system $\cH^{*\dagger}$ is
the greater simplicity of its associated deduction system.

We denote the set of all $\cS$-formulas by $\cS$, and similarly for
the other languages considered here. Where the language is clear from
context, we speak simply of {\em formulas}.  If $\phi = \forall(e,f)$,
we write $\bar{\phi}$ to denote $\exists(e,\bar{f})$; and if $\phi =
\exists(e,f)$, we write $\bar{\phi}$ to denote
$\forall(e,\bar{f})$. Thus, $\bar{\bar{\phi}} = \phi$, and, in any
structure $\fA$, $\fA \models \phi$ if and only if $\fA \not
\models \bar{\phi}$. It is simple to check that, if $\cL$ is
any of the languages $\cS$, $\cS^\dagger$, $\cH$ or $\cH^\dagger$,
then $\phi \in \cL$ implies $\bar{\phi} \in \cL$: that is, all the
languages introduced above are, in effect, closed under negation.  If
$\Theta$ is a set of formulas, we write $\fA \models \Theta$ if, for
all $\theta \in \Theta$, $\fA \models \theta$.  A formula $\theta$ is
{\em satisfiable} if there exists a structure $\fA$ such that $\fA
\models \theta$; a set of formulas $\Theta$ is {\em satisfiable} if
there exists $\fA$ such that $\fA \models \Theta$.  If, for all
structures $\fA$, $\fA \models \Theta$ implies $\fA \models \theta$,
we say that $\Theta$ {\em entails} $\theta$, and write $\Theta \models
\theta$. We call a formula of the form $\exists(e,\bar{e})$ an {\em
  absurdity}, and use $\bot$ to denote, indifferently, any
absurdity. Evidently, $\bot$ is unsatisfiable.

We illustrate the logics $\cH$ and $\cH^\dagger$ with some sample
entailments.  In the former case, we have, for example:
\begin{equation}
\set{\exists(p,\forall q), \exists(q,o)} \models \forall(q, o).
\label{eq:validityH}
\end{equation}
For suppose that some $p$ is identical to every $q$, and there is a $q$
which is also an $o$.  Then there is exactly one $q$, and it is $o$;
therefore, every $q$ is an $o$. In the latter case, we have, for example:
\begin{align}
& \set{\forall(p,\forall p), \forall(\bar{p}, p), \exists(q_1, q_1)}
  \models \forall(q_2, q_1)
\label{eq:validityHdagger1}\\
& \set{\forall(p,\forall p), \forall(\bar{p},\forall \bar{p}), 
\exists(q_1, \bar{q}_2), \exists(q_2, \bar{q}_3)} \models 
\forall(q_3, q_1).
\label{eq:validityHdagger2}
\end{align}
The validity~\eqref{eq:validityHdagger1} follows from the fact that
any model of the premises has a 1-element domain. Likewise,
in~\eqref{eq:validityHdagger2}, any model of the premises has a
2-element domain.  Thus, in the language $\cH^\dagger$, it is possible to
write satisfiable sets of formulas whose only models are of size 1 or
2. (This is trivially impossible in $\cH$.) On the other hand, we
shall see in Theorem~\ref{theo:largeModels} that, if a set of
$\cH^{*\dagger}$-formulas has a model with three or more elements,
then it has arbitrarily large models.

To ease readability in proofs, we employ the following variable-naming
conventions.  The variables $o$, $p$ and $q$ (possibly with decorations)
are to be understood as ranging only over atoms, the variables $\ell$
and $m$ only over literals, the variables $c$ and $d$ only over c-terms,
and the variables $e$, $f$ and $g$ over e-terms. Thus, for example, if
$S$ is a set of e-terms, the statement ``there exists $\ell \in S$
\ldots'' should be read as ``there exists a literal $\ell \in S$
\ldots'', and so on.

\section{Proof Theory}
\label{sec:proofTheory}
By a {\em syllogistic language}, we mean any of the languages $\cS$,
$\cS^\dagger$, $\cH$, $\cH^\dagger$ or $\cH^{*\dagger}$. This
enumerative definition could be replaced by a more general
characterization; however, the details are not relevant to the concerns
of this paper, and we avoid them.  The problem of finding sound and
complete rule-systems for the language $\cS$ was solved
(independently) in Smiley~\cite{qpr:smiley73},
Corcoran~\cite{qpr:corcoran} and Martin~\cite{qpr:martin97}. This
result is strengthened marginally in Pratt-Hartmann and
Moss~\cite{qpr:p-h+m09} (as explained below), and extended to the
language $\cS^\dagger$.  Here, we seek a system of syllogistic rules
which generate exactly the entailments in the languages $\cH$,
$\cH^\dagger$ and $\cH^{*\dagger}$. Because our results will be partly
negative in character, we adopt a relatively formal approach.

Let $\cL$ be a syllogistic language.  A {\em syllogistic rule} in
$\cL$ is a pair $\Theta/\theta$, where $\Theta$ is a finite set
(possibly empty) of $\cL$-formulas, and $\theta$ an $\cL$-formula.  We
call $\Theta$ the {\em antecedents} of the rule, and $\theta$ its {\em
  consequent}.  We generally display rules in `natural-deduction'
style. For example,
\begin{equation}
\begin{array}{ll}
\infer{\exists (p,o)}
                   {\exists (p,q) & 
             \qquad \forall (q,o)}    
\hspace{3cm} & 
\infer[,]{\exists (p,\bar{o})}
                   {\exists (p,q) & 
             \qquad \forall (q,\bar{o})}    
\end{array}
\label{eq:dariiFerio}
\end{equation}
where $p$, $q$ and $o$ are atoms, are syllogistic rules in $\cS$,
(hence in any larger syllogistic language) corresponding to the
traditional syllogisms {\em Darii} and {\em Ferio}, respectively. A
rule is {\em valid} if its antecedents entail its consequent. Thus, the
rules~\eqref{eq:dariiFerio} are valid.  As a further example, the
following obvious generalizations of~\eqref{eq:dariiFerio} are valid
syllogistic rules in $\cH$:
\begin{equation}
\begin{array}{ll}
\infer{\exists (p,\forall o)}
                   {\exists (p,q) & 
             \qquad \forall (q, \forall o)}    
\hspace{3cm} & 
\infer[.]{\exists (p,\overline{\forall o})}
                   {\exists (p,q) & 
             \qquad \forall (q,\overline{\forall o})}
\end{array}
\label{eq:nonClassical}
\end{equation}

Let $\cL$ be a syllogistic language and $\sX$ a set of syllogistic
rules in $\cL$; and denote by $\bbP(\cL)$ the set of subsets of
$\cL$. A {\em substitution} is a function $g: \bP \rightarrow \bP$; we
extend $g$ to $\cL$-formulas and to sets of $\cL$-formulas in the
obvious way.  An {\em instance} of a syllogistic rule $\Theta/\theta$
is the syllogistic rule $g(\Theta)/g(\theta)$, where $g$ is a
substitution.  We define the {\em direct syllogistic derivation
  relation} $\vdash_\sX$ to be the smallest relation on $\bbP(\cL)
\times \cL$ satisfying:
\begin{enumerate}
\item if $\theta \in \Theta$, then $\Theta \vdash_\sX \theta$;
\item if $\{\theta_1, \ldots, \theta_n\}/\theta$ is a syllogistic rule
  in $\sX$, $g$ a substitution, $\Theta = \Theta_1 \cup \cdots \cup
  \Theta_n$, and $\Theta_i \vdash_\sX g(\theta_i)$ for all $i$ ($1
  \leq i \leq n$), then $\Theta \vdash_\sX g(\theta)$.
\end{enumerate}
Where the language $\cL$ is clear from context, we omit
reference to it; further, we typically contract {\em syllogistic rule}
to {\em rule}.  Instances of the relation $\vdash_{\sX}$ can
always be established by {\em derivations} in the form of finite trees
in the usual way. For instance, the derivation
\begin{equation*}
\infer[(\mbox{\small D1})]{\exists(p,\bar{r})}
        {\infer[(\mbox{\small D1})]{\exists(p,o)}
               {\exists(p,q)
                &
                \forall(q,o)}
        &
        \forall(o,\bar{r})}
\end{equation*}
establishes that, for any set of syllogistic rules ${\sX}$
containing the rules~\eqref{eq:dariiFerio}, 
\begin{equation*}
\set{\exists(p,q), \forall(q,o), \forall(o,\bar{r})} \vdash_{\sX} 
   \exists(p,\bar{r}).
\end{equation*}
In the sequel, we reason freely about derivations in order to
establish properties of derivation relations. The tags (D1) merely
serve to indicate the rule employed in each step of the derivation:
both the rules in~\eqref{eq:dariiFerio} fall under a group which we shall
later call (D1).

The syllogistic derivation relation $\vdash_\sX$ is said to be {\em
  sound} if $\Theta \vdash_\sX \theta$ implies $\Theta \models
\theta$, and {\em complete} (for $\cL$) if $\Theta \models \theta$
implies $\Theta \vdash_\sX \theta$.  A set $\Theta$ of formulas is
\emph{inconsistent} ({\em with respect to} $\vdash_\sX$) if $\Theta
\vdash_\sX \bot$ for some absurdity $\bot$; otherwise,
\emph{consistent}. It is obvious that, for any set of rules $\sX$,
$\vdash_\sX$ is sound if and only if every rule in $\sX$ is valid. A
weakening of completeness called \emph{refutation-completeness} will
prove important in the sequel: $\vdash_\sX$ is {\em
  refutation-complete} if any unsatisfiable set $\Theta$ is
inconsistent with respect to $\vdash_\sX$.  Completeness trivially
implies refutation-completeness, but not conversely.

The languages $\cH^\dagger$ and $\cH^{*\dagger}$ turn out to require a
stronger form of proof-system than that provided by direct derivation
relations.  Let $\cL$ be a syllogistic language and $\sX$ a set of
syllogistic rules in $\cL$. We define the {\em indirect syllogistic
  derivation relation} $\Vdash_\sX$ to be the smallest relation on
$\bbP(\cL) \times \cL$ satisfying:
\begin{enumerate}
\item if $\theta \in \Theta$, then $\Theta \Vdash_\sX \theta$;
\item if $\{\theta_1, \ldots, \theta_n\}/\theta$ is a syllogistic rule
  in $\sX$, $g$ a substitution, $\Theta = \Theta_1 \cup \cdots \cup
  \Theta_n$, and $\Theta_i \Vdash_\sX g(\theta_i)$ for all $i$ ($1
  \leq i \leq n$), then $\Theta \Vdash_\sX g(\theta)$.
\item if
  $\Theta \cup \set{\theta} \Vdash_\sX \bot$, where $\bot$ is any absurdity,
  then $\Theta \Vdash_\sX \bar{\theta}$.
\end{enumerate}
The only difference is the addition of the final clause, which allows
us to derive a formula $\bar{\theta}$ from premises $\Theta$ if we can
derive an absurdity from $\Theta$ together with $\theta$.  Instances
of the indirect derivation relation $\Vdash_{\sX}$ may also be
established by constructing derivations, except that we need a little
more machinery to keep track of premises. This may be done as
follows. Suppose we have a derivation (direct or indirect) showing
that $\Theta \cup \{\theta\} \Vdash_{\sX} \bot$, for some absurdity
$\bot$. Let this derivation be displayed as
\begin{equation*}
\infer*{\hspace{1mm} \bot,} 
       {\theta_1 & \cdots & \theta_n & \theta & \cdots & \theta}
\end{equation*}
where $\theta_1, \ldots, \theta_n$ is a list of formulas of $\Theta$
(not necessarily exhaustive, and with repeats allowed).
Applying Clause~3 of the definition of $\Vdash_{\sX}$, we have
$\Theta \Vdash_{\sX} \bar{\theta}$, which we take to be established
by the derivation
\begin{equation*}
        \infer[\mbox{\small (RAA)}^1.]{\bar{\theta}} {\hspace{2mm}  \infer*{\bot}
          {\theta_1 &  \cdots &  \theta_n & [\theta]^1
            \cdots & [\theta]^1} \hspace{2mm}}
\end{equation*}
The tag (RAA) stands for {\em reductio ad absurdum}; the square
brackets indicate that the enclosed instances of $\theta$ have been
{\em discharged}, i.e.~no longer count among the premises; and the
numerical indexing is simply to make the derivation history clear.
Note that there is nothing to prevent $\theta$ from occurring among
the $\theta_1, \ldots, \theta_n$; that is to say, we do not have to
discharge all (or indeed any) instances of the premise $\theta$ if we
do not want to. Again, it should be obvious that, for any set of rules
${\sX}$, $\Vdash_\sX$ is sound if and only if every rule in $\sX$ is
valid, and $\Vdash_{\sX}$ is complete if it is refutation complete. It
is important to understand that {\em reductio ad absurdum} cannot be
formulated as a syllogistic rule in the technical sense defined here;
rather, it is part of the proof-theoretic machinery that converts any
set of rules ${\sX}$ into the derivation relation $\Vdash_{\sX}$.

Syllogistic rules that differ only by renaming of atoms have the same
sets of instances, and so may be regarded as identical. That is, in
rules such as~\eqref{eq:dariiFerio} and~\eqref{eq:nonClassical}, we
may informally think of the atoms $o$, $p$ and $q$ as metavariables
ranging over the set of atoms $\bP$. This suggests the following
notational convention. Taking the meta-variable $c$ to range over
c-terms, we may comprehend the four rules in~\eqref{eq:dariiFerio}
and~\eqref{eq:nonClassical} under the single schema
\begin{equation*}
\infer[(\mbox{\small D1}).]{\exists(p, c)}
                        {\exists(p,q) \hspace{0.25cm} \forall(q, c)}
\end{equation*}
We shall employ this schematic notation in the sequel.  Note, however,
that such schemata are always shorthand for a {\em finite} number of
rules. In the sequel, we generally refer to rule schemata simply as
{\em rules}.

The following complexity-theoretic observations on derivation
relations will prove useful in this paper.
\begin{lemma}
Let $\cL$ be a syllogistic language, $\theta \in \cL$ and $\Theta
\subseteq \cL$.  If there is a derivation \textup{(}direct or
indirect\textup{)} of $\theta$ from $\Theta$ using some set of rules
$\sX$, then there is such a derivation involving only the atoms
occurring in $\Theta \cup \set{\theta}$.
\label{lma:substitute}
\end{lemma}
\begin{proof}
Given a derivation of $\theta$ from $\Theta$, uniformly replace any
unary atom that does not occur in $\Theta \cup \set{\theta}$ with one
that does.
\end{proof}
\begin{proposition}
Let $\cL$ be a syllogistic language, and $\sX$ a finite set of
syllogistic rules in $\cL$. The problem of determining whether $\Theta
\vdash_\sX \theta$, for a given set of $\cL$-formulas $\Theta$ and
$\cL$-formula $\theta$, is in \PTIME. Hence, if $\vdash_\sX$ is sound
and refutation-complete, the satisfiability problem for $\cL$ is in
\PTIME.
\label{prop:PTIME}
\end{proposition}
\begin{proof}
By Lemma~\ref{lma:substitute}, we may confine attention to derivations
featuring only the atoms in $\Theta \cup \set{\theta}$.  The number
$m$ of $\cL$-formulas featuring these atoms is bounded by a quadratic
function of $|\Theta \cup \set{\theta}|$; evidently, we need only
consider derivations with $m$ or fewer steps. Suppose that the maximum
number of premises of any rule in $\sX$ is $\ell$.  If $k < m$, and
the set of formulas derivable in $k$ steps has been computed, then we
may evidently compute the number of formulas derivable in $k+1$ steps
in time $O(m^{\ell+1})$. 
\end{proof}
Note that Proposition~\ref{prop:PTIME} does not apply to indirect
derivation relations. However, we do have a weaker global complexity
bound, even in this case.  If $\Phi$ is set of $\cL$-sentences, we say
that $\Phi$ is {\em complete} if, for every $\cL$-sentence $\phi$
featuring only the atoms occurring in $\Phi$, either $\phi \in \Phi$
or $\bar{\phi} \in \Phi$. Trivially, every satisfiable set of
$\cL$-sentences can be extended to a complete, satisfiable set of
$\cL$-sentences.  (Do not confuse this observation with
Lemma~\ref{lma:lindenbaum}.)
\begin{proposition}
Let $\cL$ be a syllogistic language, $\sX$ a finite set of syllogistic
rules in $\cL$, and $\Psi$ a complete set of $\cL$-sentences.  If
$\Psi \Vdash_\sX \bot$, then $\Psi \vdash_\sX \bot$.  Hence, if
$\Vdash_\sX$ is sound and complete, the satisfiability problem for
$\cL$ is in \NPTIME.
\label{prop:NPTIME}
\end{proposition}
\begin{proof}
For the first statement, suppose that there is an indirect derivation
of some absurdity $\bot$ from $\Psi$, using the rules $\sX$. Let the
number of applications of (RAA) employed in this derivation be $k$;
and assume without loss of generality that $\bot$ is chosen so that
this number $k$ is minimal. If $k > 0$, consider the last application
of (RAA) in this derivation, which derives a formula, say,
$\bar{\psi}$, discharging a premise $\psi$. Then there is an
(indirect) derivation of some absurdity $\bot'$ from $\Psi \cup
\{\psi\}$, employing fewer than $k$ applications of (RAA).  By
minimality of $k$, $\psi \not \in \Psi$, and so, by the completeness
of $\Psi$, $\bar{\psi} \in \Psi$. But then we can replace our original
derivation of $\bar{\psi}$ with the trivial derivation, so obtaining a
derivation of $\bot$ from $\Psi$ with fewer than $k$ applications of
(RAA), a contradiction. Therefore, $k = 0$, or, in other words, $\Psi
\vdash_\sX \bot$.  For the second statement, let a set of
$\cL$-sentences $\Phi$ be given. Now guess a complete superset $\Psi$
involving only those atoms occurring in $\Phi$.  Evidently, $|\Psi|$
is bounded by a polynomial function of $|\Phi|$.  By
Proposition~\ref{prop:PTIME}, we can check in polynomial time whether
$\Psi \vdash_\sX \bot$.
\end{proof}

We mentioned above that the existence of sound and complete
syllogistic systems for the languages $\cS$ and $\cS^\dagger$ has been
solved. More specifically, it is shown in Pratt-Hartmann and
Moss~\cite{qpr:p-h+m09}, that, for both languages, a finite set of
rules exist for which the associated direct derivation relation is
sound and complete. (The earlier work cited above showed only the
existence of sound and refutation-complete systems for $\cS$.) We are
now in a position to state the technical results of this paper:
\begin{enumerate}
\item
There is no finite set $\sX$ of syllogistic rules in $\cH$ such that
$\vdash_\sX$ is sound and complete (Theorem~\ref{theo:noComplete}).  
\item
There is a finite set ${\sH}$ of syllogistic rules in $\cH$ such that
$\vdash_{\sH}$ is sound and refutation-complete
(Theorem~\ref{theo:refComplete}).
\item
The problem of determining whether a set of $\cH^\dagger$-formulas is
satisfiable is \NPTIME-complete, and similarly for the problem of
determining whether a set of $\cH^{*\dagger}$-formulas is satisfiable
(Theorem~\ref{theo:NPcomplete}).  Hence, by Proposition~\ref{prop:PTIME},
unless \PTIME=\NPTIME, there
is no finite set $\sX$ of syllogistic rules in either $\cH^\dagger$ or
$\cH^{*\dagger}$ such that $\vdash_\sX$ is sound and
refutation-complete.
\item There is a finite set ${\sH^\dagger}$ of syllogistic rules in
  $\cH^\dagger$ such that $\Vdash_{\sH^\dagger}$ is sound and complete
  (Theorem~\ref{theo:completeDagger}).
\item
There is a finite set ${\sH^{*\dagger}}$ of syllogistic rules in
$\cH^{*\dagger}$ such that $\Vdash_{\sH^{*\dagger}}$ is sound and
complete (Theorem~\ref{theo:completeStarDagger}).
\end{enumerate}
The following sections of this paper are devoted to proofs of these
results. We round of the present section by establishing a version of
the Lindenbaum Lemma for indirect derivation relations. This result
will be used in Section~\ref{sec:complete}.
\begin{lemma}
Let $\cL$ be a syllogistic language, $\sX$ a finite set of syllogistic
rules in $\cL$, and $\Phi$ a set of $\cL$-formulas. If $\Phi$ is
$\Vdash_\sX$-consistent, then $\Phi$ has a $\Vdash_\sX$-consistent,
complete extension.
\label{lma:lindenbaum}
\end{lemma}
\begin{proof}
Enumerate the $\cL$-formulas as $\phi_0, \phi_1, \ldots$. Define
$\Phi_0 = \Phi$, and
\begin{equation*}
\Phi_{i+1} = 
\begin{cases}
\Phi \cup \set{\phi_{i}} & \text{if $\Phi \not \Vdash_\sX \bar{\phi}_i$}\\
\Phi \cup \set{\bar{\phi}_{i}} & \text{otherwise},
\end{cases}
\end{equation*}
for all $i \geq 0$.  We show by induction that each $\Phi_i$ is
consistent. From this it follows that $\Phi^* = \bigcup_{0 \leq i}
\Phi_i$ is consistent, thus proving the lemma. The case $i = 0$ is
true by hypothesis; so we suppose that $\Phi_i$ is consistent, but
$\Phi_{i+1}$ inconsistent, and derive a contradiction. Assume first
that $\Phi_i \not \Vdash_\sX \bar{\phi}_i$. Thus, $\Phi_{i+1} = \Phi
\cup \set{\phi_i} \Vdash_\sX \bot$, whence, by the rule (RAA), $\Phi_i
\Vdash_\sX \bar{\phi_i}$, contrary to assumption. 
On the other hand, assume $\Phi_i \Vdash_\sX
\bar{\phi}_i$, so that $\Phi_{i+1} = \Phi_i \cup \set{\bar{\phi}_i}$.
Take derivations establishing that $\Phi_i \Vdash_\sX
\bar{\phi}_i$ and that 
$\Phi_i \cup \set{\bar{\phi}_i} \Vdash_\sX \bot$; and chain these together
to form a single derivation, thus:
\begin{equation*}
\infer*{\bot}
   {\Phi_i,  \infer*{\bar{\phi}_i}{\Phi_i}}.
\end{equation*}
This establishes that $\Phi_i \Vdash_\sX \bot$, contrary to the
supposed consistency of $\Phi$.
\end{proof}

\section{No complete syllogistic systems for $\cH$}
\label{sec:noComplete}
The objective of this section is to prove
\begin{theorem}
There is no finite set $\sX$ of syllogistic rules in $\cH$ such that
$\vdash_\sX$ is sound and complete.
\label{theo:noComplete}
\end{theorem}
We use a variant of a technique from Pratt-Hartmann and
Moss~\cite{qpr:p-h+m09}. For $n \geq 3$, let $\Gamma^n$ be the
set of formulas
\begin{align}
& \forall(p_i, \overline{\forall p_{i+1}}) & & (1 \leq i < n)
\label{eq:completeness1}\\
& \forall(p_1, \forall p_n) & & 
\label{eq:completeness2}\\
& \forall(p_n, \forall p_1) & & 
\label{eq:completeness2a}\\
& \forall(p_i,p_i) & & (1 \leq i \leq n)
\label{eq:completeness3}\\
& \forall(p_1,\bar{p}_{n-1}) & &
\label{eq:completeness4}
\end{align} 
and let $\gamma^n$ be the formula $\forall(p_1, p_n)$. Note that the
Formulas~\eqref{eq:completeness2} and~\eqref{eq:completeness2a} are
logically equivalent, that Formulas~\eqref{eq:completeness3} are true
in every structure, and that Formula~\eqref{eq:completeness4} is an
immediate consequence of~\eqref{eq:completeness1} (putting $i = n-1$)
and~\eqref{eq:completeness2}. Further, $\Gamma^n \models \gamma^n$. To
see this, suppose for contradiction that $\fA \models \Gamma^n$, but
$a \in p_1^\fA \setminus p_n^\fA$. Since $p_1^\fA \neq \emptyset$, the
formulas~\eqref{eq:completeness1} ensure that $p_i^\fA \neq \emptyset$
for all $i$ ($1 \leq i \leq n$).  By~\eqref{eq:completeness2}, then,
$a$ is the unique element of $p_n^\fA$.  But this contradicts the fact
that $a \not \in p_n^\fA$.  We proceed to show that, for any finite
set $\sX$ of syllogistic rules, if $\vdash_\sX$ is sound, then there
exists a value of $n$ such that $\Gamma^n \not \vdash_\sX \gamma^n$.

For any $h$, $1 \leq h \leq n-2$, define $\Gamma^n_h  = \Gamma^n \setminus
\set{\forall(p_h,\overline{\forall p_{h+1}})}$.
\begin{lemma}
Let $\phi$ be an $\cH$-formula featuring only the atoms $p_1, \ldots,
p_n$, and let $1 \leq h \leq n-2$. Then either $\Gamma^n_h \not \models
\phi$ or $\phi \in \Gamma^n$.
\label{lma:c1}
\end{lemma}
\begin{proof}
We consider the possible forms of $\phi$ in turn.

\bigskip

\noindent
1. $\phi = \forall(p_i,p_j)$: 
Let $A = \set{a_1, \ldots, a_{n-1}}$, and define the structure
$\fA$ over $A$ by setting
\begin{equation*}
p_k^\fA = \set{a_k} \quad  (1 \leq k < n), \qquad p_n^\fA  =  \set{a_1}.
\end{equation*}
A routine check shows that $\fA \models \Gamma^n_h$, but, for $i \neq
j$ and $\set{i,j} \neq \set{1,n}$, $\fA \not \models \phi$. On the
other hand, if $i = j$ then, from~\eqref{eq:completeness3}, $\phi \in
\Gamma^n$. This means we need only deal with the case $\set{i,j} =
\set{1,n}$.  For all $h$ ($1 \leq h \leq n-2$), let $C_h = \set{a_1,
  \ldots, a_h}$ and $D_h = \set{a_{h+1}, \ldots, a_n}$, and define the
structures $\fC_h$ and $\fD_h$ by setting:
\begin{align*}
p_k^{\fC_h} &= \set{a_k} & & (1 \leq k \leq h) & p_k^{\fC_h} & = \emptyset & & (h < k \leq n)\\ 
p_k^{\fD_h} &= \emptyset & & (1 \leq k \leq h) & p_k^{\fD_h} & = \set{a_k} & & (h < k \leq n).
\end{align*}
A routine check shows that $\fC_h \models \Gamma^n_h$ and $\fD_h
\models \Gamma^n_h$, but that $\fC_h \not \models \forall(p_1,p_n)$
and $\fD_h \not \models \forall(p_n,p_1)$.

\bigskip

\noindent
2. $\phi = \exists(p_i,c)$: It is immediate that
\begin{align*}
\fC_h \not \models \exists(p_i, c) & & 
     \text{($h < i < n$ and $c$ any c-term)}\\
\fD_h \not \models \exists(p_i, c) & & 
     \text{($1 \leq i \leq h$ and $c$ any c-term)}
\end{align*}
Hence, if $i \leq h$, $\fD_h \models \Gamma^n_h$, but $\fD_h \not
\models \phi$; if $i > h$, $\fC_h \models \Gamma^n_h$, but $\fC_h \not
\models \phi$.

\bigskip

\noindent
3. $\phi = \forall(p_i,\bar{p}_j)$: Given that the formulas
$\forall(p_i,\bar{p}_j)$ and $\forall(p_j,\bar{p}_i)$ are identified
in this paper, we may assume without loss of generality that $i \leq
j$.  If $i =j$, or if $i = 1$ and $j = n$, then $\fA \not \models
\phi$. If $i = 1$ and $j = n-1$, then, from~\eqref{eq:completeness4},
we have $\phi \in \Gamma^n$. We next suppose that either $1 \leq i < j
\leq n -2$, or $2 \leq i < j \leq n -1$. For such values of $i$ and
$j$, define $\fA_{i,j}$ over $A$ by setting
\begin{equation*}
p_k^\fA = \set{a_k} \quad  (1 \leq k < n) \mbox{ and } k \neq i, \qquad
p_i^\fA = \set{a_i, a_j}, \qquad p_n^\fA  =  \set{a_1}.
\end{equation*}
Thus, $\fA_{i,j}$ is just like $\fA$, except that the element $a_i$
additionally realizes the predicate $p_j$.  A routine check shows that
$\fA_{i,j} \models \Gamma^n_h$, but $\fA_{i,j} \not \models
\phi$. (Note that $\fA_{i,j}$ is not defined if $j = n$ or if $i = 1$
and $j = n-1$.) We next suppose that $2 \leq i \leq n -2$ and $j = n$.
Define $\fA_i$ to be just like $\fA$, except that the first element
$a_1$ additionally realizes the predicate $p_i$. A routine check
shows that $\fA_i \models \Gamma^n_h$, but $\fA_i \not \models \phi$.
(Note that $\fA_i$ is not defined if $i=1$ or $n-1 \leq i \leq n$.)  The only
remaining case is where $i = n-1$ and $j = n$.  Define the structure
$\fD'_h$ to be just like $\fD_h$, except that $a_{n-1}$
additionally satisfies the predicate $p_n$. Again, a routine check
shows that $\fD'_h \models \Gamma^n_h$, but $\fD \not \models \phi$.

\bigskip

\noindent
4. $\phi = \forall(p_i,\forall p_j)$: Given that
Formulas~\eqref{eq:completeness2} and~\eqref{eq:completeness2a} are
the only formulas of this form in $\Gamma^n$, we may assume without
loss of generality that $i \leq j$, and also that either $1 < i$ or $j
< n$.  If $1 \leq i < j \leq n$ and $\set{i,j} \neq \set{1,n}$, then
$\fA \not \models \phi$. This leaves only the case where $i =
j$. Denote by $2 \times \fC_h$ the resulting of taking two disjoint
copies of $\fC_h$, and similarly for $2\times \fD_h$. A routine check
shows that $2 \times \fC_h \models \Gamma^n_h$ and $2 \times \fD_h
\models \Gamma^n_h$. On the other hand
\begin{align*}
2 \times \fC_h \not \models \forall(p_i, \forall p_i) & & 
     \text{$1 \leq i \leq h$}\\
2 \times \fD_h \not \models \forall(p_i, \forall p_i) & & 
     \text{$h < i < n$.}
\end{align*}
Thus, if $i = j \leq h$, then $2 \times \fC_h \not \models \phi$,
and if $i = j > h$, then $2 \times \fD_h \not \models \phi$.

\bigskip

\noindent
5. $\phi = \forall(p_i,\overline{\forall p_j})$: If $i = j$, then $\fA
\not \models \phi$. If $j = i+1$, then, from~\eqref{eq:completeness1},
$\phi \in \Gamma^n$. If $1 \leq i \leq h < j \leq n$, then $\fC_h \not
\models \phi$.  If $1 \leq j \leq h < i \leq n$, then $\fD_h \not
\models \phi$.  If $1 \leq i \leq h$, $1 \leq j \leq h$ and $j \neq
i+1$, let the structure $\fC_{h,i,j}$ be just like $\fC_h$, except that
$a_j$ additionally satisfies the predicate $p_i$.  A routine check
shows that $\fC_{h,i,j} \models \Gamma^n_h$, but $\fC_{h,i,j} \not
\models \phi$. (Note that $\fC_{h,i,j}$ is not defined if $j = i+1$.)
Similarly, if $h < i \leq n$, $h < j \leq n$ and $j \neq i+1$, let the
structure $\fD_{h,i,j}$ be just like $\fD_h$, except that $a_j$
additionally satisfies the predicate $p_i$.  Again, we have
$\fD_{h,i,j} \models \Gamma^n_h$, but $\fD_{h,i,j} \not \models \phi$.
\end{proof}
\begin{proof}[Proof of Theorem~\ref{theo:noComplete}]
Let $\sX$ be a finite (non-empty) set of syllogistic rules such that
$\vdash_\sX$ is sound.  Let the maximum number of antecedents in any
of the rules of $\sX$ be $r \geq 0$, fix $n = r+3$, and let $\theta$
be any $\cH$-formula featuring only the atoms $p_1, \ldots, p_n$.  We
claim that $\Gamma^{(n)} \vdash_\sX \theta$ implies $\theta \in
\Gamma^{n}$. Since $\gamma^n \not \in \Gamma^n$, this proves the
theorem.

\bigskip

\noindent
We prove the claim by induction on the lengths of derivations.  By
Lemma~\ref{lma:substitute}, if there is a derivation of $\gamma^n$
from $\Gamma^n$, then there is such a derivation using only the atoms
$p_1, \ldots, p_n$. Henceforth, then, we confine ourselves to
derivations featuring only these atoms. Now, for derivations employing
no steps of inference---i.e.~for $\theta \in \Gamma^n$---the claim is
trivial.  So suppose that the claim holds for derivations employing at
most $p$ steps, and that $\theta$ is derived from $\Gamma$ in $p+1$
steps. By inductive hypothesis, the antecedents of the final
rule-instance will all be in $\Gamma^n$; therefore, since $n = r+3$,
the antecedents of the final rule-instance will all be in
$\Gamma_h^n$, for some $h$ ($1 \leq h \leq n-2$). Since $\vdash_\sX$
is sound, $\Gamma_h^n \models \theta$, whence, by Lemma~\ref{lma:c1},
$\theta \in \Gamma^n$.  This completes the inductive step, and the
proof of the theorem.
\end{proof}
\section{A refutation-complete syllogistic system for $\cH$}
\label{sec:refutationComplete}
The objective of this section is to prove
\begin{theorem}
There is a finite set ${\sH}$ of syllogistic rules in $\cH$ such that
the direct derivation relation $\vdash_{\sH}$ is sound and
refutation-complete.
\label{theo:refComplete}
\end{theorem}
We display $\sH$ in schematic form, with $o$, $p$ and $q$ ranging over
atoms, and $c$ over $c$-terms, as usual. The rule-schemata fall naturally
into four groups.
\begin{enumerate}
\item `little' rules:
\begin{equation*}
\infer[(\mbox{\small I})]{\exists(p,p)}{\exists(p,c)}
\hspace{1cm}
\infer[(\mbox{\small T});]{\forall(p,p)}{}
\end{equation*}
\item rules similar to familiar syllogisms:
\begin{equation*}
\begin{array}{lll}
% Barbara (1st figure AAA)
\infer[(\mbox{\small B})]{\forall(p, c)}
                        {\forall(p,q) \hspace{0.25cm} \forall(q,c)}
\\
\\
\infer[(\mbox{\small D1})]{\exists(p, c)}
                        {\exists(p,q) \hspace{0.25cm} \forall(q, c)}
\hspace{0.2cm}
& 
\infer[(\mbox{\small D2})]{\exists(q, c)}
                        {\exists(p,c) \hspace{0.25cm} \forall(p, q)}
\hspace{0.2cm}
&
\infer[(\mbox{\small D3});]{\exists(p, \bar{q})}
                        {\exists(p,c) \hspace{0.25cm} \forall(q, \bar{c})}
\end{array}
\end{equation*}
\item `little' rules for universally quantified predicates:
\begin{equation*}
\infer[(\mbox{\small H1})]{\forall(q,p)}{\exists(p,\forall q)}
\hspace{1cm}
\infer[(\mbox{\small H2})]{\forall(q,\forall q)}{\exists(p,\forall q)}
\hspace{1cm}
\infer[(\mbox{\small H3});]{\exists(q,\overline{\forall p})}
                           {\exists(p,\overline{\forall q})}
\end{equation*}
\item syllogism-like rules for universally quantified predicates
\begin{equation*}
\begin{array}{ll}
\infer[(\mbox{\small HH1})]{\forall(q, c)}
                        {\exists(q,c) & \exists(p,\forall q)}
\hspace{1cm} &
\infer[(\mbox{\small HH2})]{\forall(q,c)}
                        {\exists(p,c) & \forall(p,\forall q)}\\
\ \\
\infer[(\mbox{\small HH3})]{\forall(q, c)}
                        {\forall(p,c) & \exists(p,\forall q)}
&
\infer[(\mbox{\small HH4}).]{\forall(p, q)}
                        {\forall(p, \forall q) & \exists(q, q)}
\end{array}
\end{equation*}
\end{enumerate}
Note that these rule-schemata define a finite set of rules, as
explained above.  Our choice of labels (I), (T), etc.~is essentially
arbitrary, though (B), (D1), (D2) and (D3) allude vaguely to the
classical syllogisms {\em Barbara} and {\em Darii}.  Recalling our
decision silently to identify the formulas $\forall(p,c)$ and
$\forall(\bar{c},\bar{p})$, (D3) could be alternatively written as
$\set{\exists(p,c), \forall(c, \bar{q})}/ \exists(p, \bar{q})$.
Validity of these rules is transparent: Rule (HH1) is a
straightforward generalization of the validity~\eqref{eq:validityH}
considered above; the other `Hamiltonian' rules are dealt with
similarly.

Let $\Phi$ be a set of $\cH$-formulas containing at least one
existential formula, such that $\Phi$ is consistent with respect to
$\vdash_\sH$.  In the following lemmas, we build a structure $\fA$,
and show that $\fA \models \Phi$.  Since $\sH$ is the only set of
rules we shall be concerned with in this section, we write $\vdash$
for the direct proof-relation $\vdash_{\sH}$.  We remind the reader
that the variables $o$, $p$ and $q$ are silently assumed to range only
over atoms, and the variables $c$ and $d$ over c-terms.

Let $S$ be a set of c-terms. We define $S^*$ to be the smallest set of
c-terms including $S$ such that, for all atoms $p$, $q$ and all
c-terms $c$:
\begin{align*}
& p \in S^* \mbox{ and } \Phi \vdash \forall(p,c) \Rightarrow c \in S^*
\tag{C1}\\
& (\forall p) \in S^* \mbox{ and } \Phi \vdash \exists(p,p) \Rightarrow p \in S^*.  
\tag{C2}
\end{align*}
Evidently, we may regard $S^*$ as the limit of a process in which,
starting with $S$, c-terms are added one by one to ensure fulfillment
of the above conditions.  More precisely, we may write $S^* =
\bigcup_{0 \leq i < \alpha} S^{(i)}$, where $S^{(0)} = S$, $\alpha
\leq \omega$, and, for all $i$ ($i+1 < \alpha$), $S^{(i+1)} = S^{(i)}
\cup \set{c}$ for some $c \not \in S^{(i)}$ satisfying either of the
following conditions:
\begin{align*}
& \mbox{there exists 
$p \in S^{(i)}$ such that $\Phi \vdash \forall(p,c)$}; 
\tag{K1}\\
& \mbox{$c = p$ is an atom such that 
   $(\forall p) \in S^{(i)}$, and $\Phi \vdash \exists(p,p)$}.
\tag{K2}
\end{align*}
Define the set $W$ to be the following set of c-terms:
\begin{eqnarray*}
W_0    & = & \set{ \set{p,c}^* \mid \Phi \vdash \exists(p,c)}\\
W_{i+1} & = & \set{ \set{p}^* \mid \overline{\forall p} \in w 
                                    \mbox{ for some } w \in W_i} \qquad (i \geq 0)\\
W      & = & \bigcup_{i \geq 0} W_i.
\end{eqnarray*}
Since $\Phi$ contains at least one existential formula, $W$ is
non-empty.  We use letters $u$, $v$, $w$ to range over elements of
$W$.  Lemmas~\ref{lma:1}--\ref{lma:6} establish some properties of
$W$.
\begin{lemma}
Let $c \in w \in W$. Then there exists $o \in w$ such that $\Phi
\vdash \exists(o,c)$.
\label{lma:1}
\end{lemma}
\begin{proof}
Assume first that $w \in W_0$. Thus, $w = \set{o,d}^*$, where $\Phi
\vdash \exists(o,d)$. Using the representation $w = \bigcup_{0 \leq i
  < \alpha} S^{(i)}$, where $S^{(0)} = \set{o,d}$, we show by
induction on $i$ that, if $c \in S^{(i)} \in W$, there exists $o \in
w$ such that $\Phi \vdash \exists(o,c)$.

\bigskip

\noindent
For $i = 0$, we have $c = d$ or $c = o$. In the former case, $\Phi
\vdash \exists(o,c)$, by assumption.  In the latter, we have the
derivation
\begin{equation*}
\infer[\mbox{\small (I)},]{\exists(o,o)}
   {\infer*{\exists(o,d)}{}}
\end{equation*}
so that, either way, $\Phi \vdash \exists(o,c)$. For $i \geq 1$, we
consider the following cases, corresponding to the
conditions~(K1)--(K2).

\bigskip

\noindent
1. $\Phi \vdash \forall(q,c)$ for some $q \in S^{(i-1)}$: By inductive
hypothesis, there exists $o \in w$ such that $\Phi \vdash
\exists(o,q)$, so we have the derivation
\begin{equation*}
\infer[{\mbox{\small (D1)}.}]{\exists(o,c)}
    {\infer*{\exists(o,q)}{} & \infer*{\forall(q,c)}{}}
\end{equation*}

\bigskip

\noindent
2. $c = q$, $\forall q \in S^{(i-1)}$, and $\Phi \vdash \exists(q,q)$:
But then there is nothing to show, since we may put $o = q$.

\bigskip

\noindent
This completes the proof of the lemma for $w \in W_0$. We now prove
the result for $w \in W_k$, for all $k \geq 0$, proceeding by
induction on $k$. For $k > 0$, we have $w = \set{o}^*$, where, for
some $v \in W_{k-1}$, $\overline{\forall o} \in v$.  By inductive
hypothesis, there exists $p \in v$ such that $\Phi \vdash \exists(p,
\overline{\forall o})$, so we have the derivation
\begin{equation*}
\infer[{\mbox{\small (I).}}]
   {\exists(o,o)}
   {\infer[{\mbox{\small (H3)}}]
       {\exists(o,\overline{\forall p})}
       {\infer*{\exists(p,\overline{\forall o})}{}}}
\end{equation*}
Having established that $\Phi \vdash \exists(o,o)$, we can proceed
exactly as for the case $k = 0$, writing $w = \bigcup_{0 \leq i <
  \alpha} S^{(i)}$, where $S^{(0)} = \set{o}$.
\end{proof}
\begin{lemma}
Let $p \in w \in W$. Then $\Phi \vdash \exists(p,p)$.
\label{lma:2}
\end{lemma}
\begin{proof}
By Lemma~\ref{lma:1}, let $o$ be such that $\Phi \vdash \exists(o,p)$.
Then we have the derivation
\begin{equation*}
\infer[\mbox{\small (I)}.]{\exists(p,p)}
      {\infer*{\exists(o,p)}{}}
\end{equation*}
\end{proof}
\begin{lemma}
If $c \in \set{p}^*$, then $\Phi \vdash \forall(p,c)$.
\label{lma:3}
\end{lemma}
\begin{proof}

Write $\set{p}^* = \bigcup_{0 \leq i < \alpha} S^{(i)}$, with $S^{(0)} =
\set{p}$, as in the proof of Lemma~\ref{lma:1}; we show that the lemma
holds for $c \in S^{(i)}$, proceeding by induction on $i$.

\bigskip

\noindent
If $i = 0$, then $c = p$, so $\Phi \vdash \forall(p,c)$ by rule~(T).
If $i \geq 1$, we again have two cases corresponding to the
conditions~(K1) and (K2).

\bigskip

\noindent
1. $\Phi \vdash \forall(q,c)$ for some $q \in S^{(i-1)}$: By inductive
hypothesis, $\Phi \vdash \forall(p,q)$, so we have the derivation
\begin{equation*}
\infer[{\mbox{\small (B)}.}]{\forall(p,c)}
    {\infer*{\forall(p,q)}{} & \infer*{\forall(q,c)}{}}
\end{equation*}

\bigskip

\noindent
2. $c = q$, $\forall q \in S^{(i-1)}$, and $\Phi \vdash \exists(q,q)$:
By inductive hypothesis, $\Phi \vdash \forall(p, \forall q)$, so we
have the derivation
\begin{equation*}
\infer[(\mbox{\small HH4}).]{\forall(p, q)}
                        {\infer*{\forall(p, \forall q)}{} & 
                         \infer*{\exists(q, q)}{}}
\end{equation*}
In both cases, $\Phi \vdash \forall(p,c)$, as required.
\end{proof}
In the next lemma, we take $\overline{\forall}$ to be the symbol $\exists$
and $\overline{\exists}$ to be the symbol $\forall$.
\begin{lemma}
Suppose $c, d \in w \in W$ with $c$, $d$ distinct. 
Then there exist $o \in w$ and $Q \in \set{\forall,\exists}$ such that
$\Phi \vdash Q(o,c)$ and $\Phi \vdash \overline{Q}(o,d)$. Hence,
if $c \in w$, then $\bar{c} \not \in w$.
\label{lma:4}
\end{lemma}
\begin{proof}
We consider first the case $w \in W \setminus W_0$. By construction of
$W$, $w = \set{o}^*$ for some atom $o$.  By Lemma~\ref{lma:2}, $\Phi
\vdash \exists(o,o)$; and by Lemma~\ref{lma:3}, $\Phi \vdash
\forall(o,c)$ and $\Phi \vdash \forall(o,d)$. But then we have the derivation:
\begin{equation*}
\infer[{\mbox{\small (D1)}.}]{\exists(o,d)}
    {\infer*{\exists(o,o)}{} & \infer*{\forall(o,d)}{}}
\end{equation*}

\noindent
Henceforth, then, we may suppose $w \in W_0$, and we again write $w =
\bigcup_{0 \leq i < \alpha} S^{(i)}$, as in the proof of Lemma~\ref{lma:1}.
Note that $S^{(0)} = \set{o,c'}$, for some atom $o$ and c-term $c'$
such that $\Phi \vdash \exists(o,c')$. We prove the lemma for $c \in
S^{(i)}$ and $d \in S^{(j)}$, proceeding by induction on $i+j$,
showing in fact that the required $o$ lies in $S^{(0)}$.

\bigskip

\noindent
If $i+j = 0$---i.e., $c, d \in S^{(0)}$---then, since $c$, $d$ are
distinct, we have $\set{c,d} = \set{o,c'}$ and $w \in W_0$.  The
result then follows immediately from the fact that, by rule (T),
$\Phi \vdash \forall(o,o)$. If $i+j >0$, assume without loss of generality
that $i > 0$.  We again have two cases corresponding to the
conditions~(K1) and (K2).

\bigskip

\noindent
1. $\Phi \vdash \forall(q,c)$ for some $q \in S^{(i-1)}$: By inductive
hypothesis, there exist $o \in S^{(0)}$ and $Q \in \set{\forall,\exists}$
such that $\Phi \vdash Q(o,q)$ and $\Phi \vdash \overline{Q}(o,d)$.
We then have one of the derivations:
\begin{equation*}
\infer[{\mbox{\small (B)}}]{\forall(o,c)}
    {\infer*{\forall(o,q)}{} & \infer*{\forall(q,c)}{}}
\hspace{2cm}
\infer[{\mbox{\small (D1)}.}]{\exists(o,c)}
    {\infer*{\exists(o,q)}{} & \infer*{\forall(q,c)}{}}
\end{equation*}
so that $\Phi \vdash Q(o,c)$, as required.

\bigskip

\noindent
2. $c = q$, $\forall q \in S^{(i-1)}$ and $\Phi \vdash \exists(q,q)$:
By inductive hypothesis, there exists $o \in S^{(0)}$ and $Q \in
\set{\forall,\exists}$ such that $\Phi \vdash Q(o, \forall q)$ and
$\Phi \vdash \overline{Q}(o,d)$. Then we have one of the
derivations
\begin{equation*}
\infer[(\mbox{\small HH4})]{\forall(o, q)}
                        {\infer*{\forall(o, \forall q)}{} & 
                         \infer*{\exists(q, q)}{}}
\hspace{2cm}
\infer[{\mbox{\small (D1)}.}]{\exists(o,q)}
    {\infer*{\exists(q,q)}{} &
     \infer[{\mbox{\small (H1)}}]{\forall(q,o)}
                                 {\infer*{\exists(o,\forall q)}{}}}
\end{equation*}
For the final statement of the lemma, suppose $c \in w$ and $\bar{c} \in w$.
Exchanging $c$ and $\bar{c}$ if necessary, let $o$ be such that
$\Phi \vdash \exists(o,c)$ and $\Phi \vdash \forall(o, \bar{c})$. Then we
have the derivation
\begin{equation*}
\infer[\mbox{(\small{D3})},]{\exists(o,\bar{o})}
  {\infer*{\exists(o,c)}{}
   &
   \infer*{\forall(o,\bar{c})}{}}
\end{equation*}
contradicting the supposed consistency of $\Phi$.
\end{proof}
\begin{lemma}
Let $p, c \in w \in W$.  Then $\Phi \vdash \exists(p,c)$.
\label{lma:5}
\end{lemma}
\begin{proof}
If $p = c$, we can apply Lemma~\ref{lma:2}. Otherwise, by
Lemma~\ref{lma:4}, let $o \in w$ and $Q \in \set{\forall,\exists}$ be
such that $\Phi \vdash Q(o,p)$ and $\Phi \vdash \overline{Q}(o,c)$. Then we
have one of the derivations
\begin{equation*}
\infer[(\mbox{\small D1})]{\exists(p,c)}
      {\infer*{\exists(p,o)}{} & \infer*{\forall(o,c)}{}}
\hspace{2cm}
\infer[(\mbox{\small D2}).]{\exists(p,c)}
      {\infer*{\exists(o,c)}{} & \infer*{\forall(o,p)}{}}
\end{equation*}
\end{proof}

\begin{lemma}
Suppose $u, v, w \in W$ with $(\forall q) \in u$, $(\forall q) \in v$ 
and $q \in w$. Then $u = v$.
\label{lma:6}
\end{lemma}
\begin{proof}
By Lemma~\ref{lma:2}, $\Phi \vdash \exists(q,q)$.  By~(C2), then, $q
\in v$.  Suppose $c \in u$, where $c \neq \forall q$.  (We already
know that $(\forall q) \in v$.)  By Lemma~\ref{lma:4}, there exists $o
\in u$ and $Q \in \set{\forall,\exists}$ such that $\Phi \vdash
Q(o,c)$ and $\Phi \vdash \overline{Q}(o,\forall q)$. Thus, we have one
of the derivations
\begin{equation*}
\infer[(\mbox{\small HH2})]{\forall(q,c)}
      {\infer*{\exists(o,c)}{} 
       & 
       \infer*{\forall(o,\forall q)}{}}
\hspace{2cm}
\infer[(\mbox{\small HH3}),]{\forall(q,c)}
      {\infer*{\forall(o,c)}{}
       & 
       \infer*{\exists(o,\forall q)}{}}
\end{equation*}
whence, by~(C1), $c \in v$. Thus, $u \subseteq v$. The reverse
inclusion follows symmetrically.
\end{proof}
Say that $w \in W$ is {\em special} if $w$ contains a c-term of the
form $\forall q$ such that $\Phi \vdash \exists(q,q)$. Intuitively,
special elements are the unique instances of some property $q$. We now
build the structure $\fA$ as follows:
\begin{eqnarray*}
A & = &\set{\langle w, 0\rangle \mid w \in W \mbox{ is special}} \cup\\
 & & \qquad
  \set{\langle w, i \rangle \mid w \in W \mbox{ is non-special, } i \in \set{-1,1}}\\ 
p^\fA & = &\set{ \langle w,i \rangle \in A \mid p \in w}, \text{ for any atom $p$.}
\end{eqnarray*}
We remark that, since $W$ is non-empty, $A$ is non-empty; so this
construction is legitimate.
\begin{lemma}
For all elements $a = \langle w, i \rangle \in A$ and all c-terms
$c$, if $c \in w$, then $a \in c^\fA$.
\label{lma:7}
\end{lemma}
\begin{proof}
We consider the possible forms of $c$ in turn.
\bigskip

\noindent
1. $c = p$ is an atom: The result is immediate by construction of $\fA$.

\bigskip

\noindent
2. $c = \bar{p}$: If, also, $p \in w$, Lemma~\ref{lma:5} guarantees
that $\Phi \vdash \exists(p,\bar{p})$, contradicting the supposed
consistency of $\Phi$. Hence, $p \not \in w$, whence, by the
construction of $\fA$, $a \in (\bar{p})^\fA$.

\bigskip

\noindent
3. $c = \forall p$: Suppose $b = \langle u, j \rangle \in A$ with $b
\in p^\fA$. By construction of $\fA$, $p \in u$, so that $\Phi \vdash
\exists(p,p)$, by Lemma~\ref{lma:2}. Furthermore, by Lemma~\ref{lma:1}, for
some $o$, $\Phi \vdash \exists(o, \forall q)$, so that we have the derivation
\begin{equation*}
\infer[(\mbox{\small H2}),]{\forall(p, \forall p)}
      {\infer*{\exists(o,\forall p)}{}}
\end{equation*}
whence $(\forall p) \in u$ by~(C1).  By
Lemma~\ref{lma:6}, $w = u$, and therefore, by construction of $A$, $i
= j = 0$. Thus, $b \in p^\fA$ implies $b = a$, so that $a \in c^\fA$.

\bigskip

\noindent
4. $c = \overline{\forall p}$: Suppose $c \in w$, and assume for the
time being that $i \neq 0$. Thus, $w$ is not special. By
Lemma~\ref{lma:1}, there exists an atom $q$ such that $\Phi \vdash
\exists(q,\overline{\forall p})$, so that we have the derivation
\begin{equation*}
\infer[{\mbox{\small (H3)}.}]{\exists (p, \overline{\forall q})}
    {\infer*{\exists (q, \overline{\forall p})}{}}
\end{equation*}
Then there exists $w' \in W_0 \subseteq W$ such that $p \in w'$; and,
by construction of $\fA$, there exists $i' \in \set{-1,0,1}$ such that
both $\langle w', i' \rangle \in q^\fA$ and $\langle w', -i' \rangle
\in p^\fA$. Since $i \neq 0$ we may suppose $i \neq i'$ (transpose $i'$
and $-i'$ if necessary), so that there exists $a' \in A$ with $a' \neq
a$ and $a' \in p^\fA$. Hence $a \in c^\fA$, as required.  Now assume
$i = 0$. Then $w$ is special, so suppose $(\forall q) \in w$, with
$\Phi \vdash \exists(q,q)$.  By (C2), $q \in w$, and by
Lemma~\ref{lma:5}, $\Phi \vdash \exists(q,\overline{\forall
  p})$. Again, therefore, by (H3), $\Phi \vdash
\exists(p,\overline{\forall q})$. By the construction of $W$, there
exists $w', \in W_0$ such that $p \in w'$ and $\overline{\forall q}
\in w'$.  By construction of $\fA$, there exists $i' \in \set{-1,0,1}$
such that $\langle w', i' \rangle \in p^\fA$.  Since $(\forall q) \in
w$ and $\overline{\forall q} \in w'$, we know from the final statement
of Lemma~\ref{lma:4} that $w \neq w'$, and therefore $a \neq a'$. Hence
$a \in c^\fA$, as required.
\end{proof}
\begin{proof}[Proof of Theorem~\ref{theo:refComplete}]
Let $\sH$ be as given above.  Soundness of $\vdash_{\sH}$ is immediate
from the fact each of these rules is valid. For
refutation-completeness, let $\Phi$ be a set of $\cH$-formulas
consistent with respect to $\vdash_\sH$. If $\Phi$ contains no
existential formulas, then $\fA \models \Phi$ for any structure $\fA$
in which $p^\fA = \emptyset$ for all $p \in \bP$. Otherwise, let $\fA$
be constructed as above. It suffices to show that $\fA \models \Phi$.
To see this, let $\phi \in \Phi$.  If $\phi = \exists(p,c)$, then, by
construction of $W$ and $A$, there exist $w \in W$ and $i \in
\set{-1,0,1}$ such that $p, c \in w$, and $a = \langle w,i \rangle \in
A$. By Lemma~\ref{lma:7}, $a \in p^\fA \cap c^\fA$ so that $\fA
\models \phi$. On the other hand, if $\phi = \forall(p,c)$, suppose $a
= \langle w,i \rangle \in p^\fA$. By construction of $\fA$, $p \in w$,
and by Condition~(C1), $c \in w'$, whence, by Lemma~\ref{lma:7}, $a
\in c^\fA$. Thus, $p^\fA \subseteq c^\fA$, so that $\fA \models \phi$.
\end{proof}
\section{\NPTIME-completeness of $\cH^\dagger$ and $\cH^{*\dagger}$}
\label{sec:NPcomplete}
The objective of this section is to prove
\begin{theorem}
The problem of determining whether a set of $\cH^\dagger$-formulas is
satisfiable is \NPTIME-complete, and similarly for the problem of
determining whether a set of $\cH^{*\dagger}$-formulas is
satisfiable. 
\label{theo:NPcomplete}
\end{theorem}
From Theorem~\ref{theo:NPcomplete} and Proposition~\ref{prop:PTIME},
it follows that, unless \PTIME{} = \NPTIME, there is no finite set $\sX$
of syllogistic rules in either $\cH^\dagger$ or $\cH^{*\dagger}$ such
that $\vdash_\sX$ is sound and refutation-complete.

Membership of these problems in \NPTIME{} is easily established by
showing that any satisfiable set $\Phi$ of $\cH^{*\dagger}$-formulas
is satisfied in a structure whose size is bounded by a polynomial
function of the number of symbols in $\Phi$. (Alternatively, the same
result is an immediate consequence of
Theorem~\ref{theo:completeStarDagger} together with
Proposition~\ref{prop:NPTIME}.)  Therefore, only the lower bounds need
be considered. We use a variant of a technique from McAllester and
Givan~\cite{logic:mcA+G92}. We remark that our task would be very easy
if we could write a set of $\cH^\dagger$-formulas whose only models
have cardinality 3. However, by Theorem~\ref{theo:largeModels}, this
is impossible.

The proof of \NPTIME-hardness proceeds by reduction of the problem
3SAT to the satisfiability problem for $\cH^\dagger$. In this context,
a {\em clause} is an expression $L_1 \vee L_2 \vee L_3$, where each
$L_k$ ($1 \leq k \leq 3$) is either a proposition letter $o$ or a
negated proposition letter $\neg o$. Given an assignment $\theta$ of
truth-values ($t$ or $f$) to proposition letters, any clause $\gamma$
receives a truth-value $\theta(\gamma)$ in the obvious way.  An
instance of the problem 3SAT is a set $\Gamma$ of clauses; that
instance is positive just in case there exists a $\theta$ such that
$\theta(\gamma) = t$ for every $\gamma \in \Gamma$.  Let $\Gamma$ be a
finite set of clauses. We show how to compute, in logarithmic space, a
set $\Phi$ of $\cH^\dagger$-formulas such that $\Phi$ is satisfiable
if and only if $\Gamma$ is a positive instance of 3SAT. To make the
proof easier, we work first with $\cH^{*\dagger}$-formulas,
strengthening the result at the very end of the proof.

First, we need formulas to represent proposition letters.  For each
proposition letter $o$ occurring in $\Gamma$, let $o_t$ and $o_f$ be
atoms (elements of $\bP$), and let $\Phi_o$ be the set of
$\cH^{*\dagger}$-formulas:
\begin{align*}
& \forall(\forall o_t, \overline{\forall o_f}) 
& &
\forall(o_t,\forall o_f) 
& &
\forall(o_t,\overline{o_f}). 
\end{align*}
Intuitively, if $\fA \models \Phi_o$, we are to interpret the
equation $o_t^\fA = \emptyset$ as stating that $o$ is true, and
$o_f^\fA = \emptyset$ as stating that $o$ is false. The following
lemma justifies this interpretation.  Suppose $\fA$ and $\fB$ are
structures and $p \in \bP$. We say that $\fA$ and $\fB$ {\em agree on}
$p$ if $p^\fA = p^\fB$. Note that if $\fA \subseteq \fB$, then $\fA$
and $\fB$ agree on $p$ just in case $p^\fB \setminus A = \emptyset$.
\begin{lemma}
If $\fA \models \Phi_o$, then $o_t^\fA = \emptyset$ if and only if
$o_f^\fA \neq \emptyset$. Conversely, suppose $A$ is a 2-element
set, and $v \in \set{t,f}$. There exists a structure $\fA_o^v$ over
$A$ such that, if $\fB \supseteq \fA_o^v$ agrees with $\fA_o^v$ on
$o_t$ and $o_f$, then $\fB \models \Phi_o$; furthermore, $(o^v)^\fB =
\emptyset$.
\label{lma:NP1}
\end{lemma}
\begin{proof}
For the first statement, suppose $\fA \models \Phi_o$.  From
$\forall(o_t,\forall o_f)$ and $\forall(o_t,\overline{o_f})$, it is
obvious that we cannot have both $o_t^\fA \neq \emptyset$ and $o_t^\fA
\neq \emptyset$. On the other hand, suppose $o_t^\fA =
\emptyset$. Then every element satisfies $\forall o_t$, and so some
element does, whence, from $\forall(\forall o_t, \overline{\forall
  o_f})$, that element is distinct from some $o_f^\fA$, so that
$o_f^\fA \neq \emptyset$.  For the second statement, let $A =
\set{a,b}$. Define the structure $\fA^t_o$ by setting $(o_t)^{\fA^t_o}
= \emptyset$ and $(o_f)^{\fA^t_o} = A$; similarly, define the structure
$\fA^f_o$ by setting $(o_f)^{\fA^f_o} = \emptyset$ and $(o_t)^{\fA^f_o}
= A$. A routine check shows that these structures have the specified
properties.
\end{proof}

Next, we need formulas to represent clauses.  For each clause $\gamma
= L_1 \vee L_2 \vee L_3 \in \Gamma$, let $s_{\gamma,1}$,
$s_{\gamma,2}$, $s_{\gamma,3}$, $s_{\gamma,4}$, $p_{\gamma,1}$,
$p_{\gamma,2}$ and $p_{\gamma,3}$ be atoms (elements of $\bP$); in
addition, let $\Phi_\gamma$ be the set of $\cH^{\dagger}$-formulas:
\begin{align*}
& \forall(s_{\gamma,1}, \forall p_{\gamma,1}) & & 
   \forall(p_{\gamma,1},s_{\gamma,2})\\ 
& \forall(s_{\gamma,2}, \forall p_{\gamma,2}) & & 
   \forall(p_{\gamma,2},s_{\gamma,3})\\ 
& \forall(s_{\gamma,3}, \forall p_{\gamma,3}) & & 
   \forall(p_{\gamma,3},s_{\gamma,4})\\ 
& \exists(s_{\gamma,1}, \bar{s}_{\gamma,4}).
\end{align*}
Intuitively, if $\fA \models \Phi_\gamma$ we are to interpret the
equation $p_{\gamma,k}^\fA = \emptyset$ as stating that $L_k$ is true
($1 \leq k \leq 3$). The next lemma justifies this interpretation.
\begin{lemma}
If $\fA \models \Phi_\gamma$, then the set of numbers $k$ \textup{(}$1 \leq k
\leq 3$\textup{)} such that $(p_{\gamma,k})^\fA = \emptyset$ is
non-empty.  Conversely, suppose $A$ is a 2-element set, and $K$ a
non-empty subset of $\set{1,2,3}$.  There exists a structure
$\fA^K_{\gamma}$ over $A$ such that, if $\fB \supseteq \fA^K_\gamma$
agrees with $\fA^K_{\gamma}$ on the atoms in $\set{s_{\gamma,1},
  s_{\gamma,2}, s_{\gamma,3}, s_{\gamma,4}, p_{\gamma,1},
  p_{\gamma,2}, p_{\gamma,3}}$, then $\fB \models \Phi_\gamma$;
furthermore, for all $k$ \textup{(}$1 \leq k \leq 3$\textup{)},
$(p_{\gamma,k})^\fB = \emptyset$ if and only if $k \in K$.
\label{lma:NP2}
\end{lemma}
\begin{proof}
For the first statement, suppose, for contradiction, that $\fA \models
\Phi_\gamma$, but $(p_{\gamma,k})^\fA \neq \emptyset$, for all $k$ ($1
\leq k \leq 3$). Since $\fA \models \exists(s_{\gamma,1},
\bar{s}_{\gamma,4})$, let $a \in s^\fA_{\gamma,1} \setminus
s^\fA_{\gamma,4}$.  Since $\fA \models \forall(s_{\gamma,1}, \forall
p_{\gamma,1})$, and neither $s_{\gamma,1}^\fA$ nor $p_{\gamma,1}^\fA$
is empty, we have $a \in p_{\gamma,1}^\fA$; moreover, since, $\fA
\models \forall(p_{\gamma,1},s_{\gamma,2})$, $a \in s^\fA_{\gamma,2}
\setminus s^\fA_{\gamma,4}$. Repeating the same reasoning twice over,
$a \in s^\fA_{\gamma,4} \setminus s^\fA_{\gamma,4}$, a contradiction.

For the second statement of the lemma, let $A = \set{a,b}$, and define
$\fA^K_\gamma$ according  to the following table.
\begin{center}
\begin{tabular}{l|l|l}
$K$ & atoms satisfied by $a$ & atoms satisfied by $b$\\
\hline
$\set{1}$ & $s_{\gamma,1}$ & $p_{\gamma,2}$, $p_{\gamma,3}$, $s_{\gamma,3}$, $s_{\gamma,4}$ \\
$\set{2}$ & $p_{\gamma,1}$, $s_{\gamma,1}$, $s_{\gamma,2}$ & $p_{\gamma,3}$, $s_{\gamma,4}$ \\
$\set{3}$ & $p_{\gamma,1}$, $p_{\gamma,2}$, $s_{\gamma,1}$, $s_{\gamma,2}$, $s_{\gamma,3}$ & - \\
$\set{2,3}$ & $p_{\gamma,1}$, $s_{\gamma_1}$, $s_{\gamma_2}$ & - \\
$\set{1,3}$ & $p_{\gamma,2}$, $s_{\gamma_1}$, $s_{\gamma_3}$ & - \\
$\set{1,2}$ & $s_{\gamma,1}$ & $p_{\gamma,3}$, $s_{\gamma,4}$ \\
$\set{1,2,3}$ & $s_{\gamma_1}$ & -
\end{tabular}
\end{center}
An exhaustive check shows that $\fA^K_\gamma$ has the required
properties.
\end{proof}

Finally, we need formulas to link proposition letters and clauses.
For each clause $\gamma = L_1 \vee L_2 \vee L_3 \in \Gamma$, and for
all $k$ ($1 \leq k \leq 3$), let the $\cH^{\dagger}$-formula
$\psi_{\gamma,k}$ be given by
\begin{equation*}
\psi_{\gamma,k} = 
\begin{cases}
\forall(o_t, \overline{\forall p_{\gamma,k}}) & \text{if $L_k = o$}\\
\forall(o_f, \overline{\forall p_{\gamma,k}}) & \text{if $L_k = \neg {o}$},
\end{cases}
\end{equation*}
and let $\Psi_\gamma = \set{\psi_{\gamma,1}, \psi_{\gamma,2}, \psi_{\gamma,3}}$.
\begin{lemma}
Suppose $\fA \models \Psi_\gamma$, and $(p_{\gamma,k})^\fA =
\emptyset$.  If $L_k = o$, then $o_t^\fA = \emptyset$; and if $L_k =
\neg {o}$, then $(o_f)^\fA = \emptyset$.
\label{lma:NP3}
\end{lemma}
\begin{proof}
Immediate.
\end{proof}
\begin{proof}[Proof of Theorem~\ref{theo:NPcomplete}]
We need only show \NPTIME-hardness. To this end, let $\Gamma$ be a set
of clauses over the proposition letters occurring $\Gamma$. Let
\begin{equation*}
\Phi = \bigcup \set {\Phi_o \mid o \mbox{ occurs in } \Gamma} \cup 
       \bigcup \set {\Phi_\gamma \cup \Psi_{\gamma} \mid \gamma \in \Gamma }.
\end{equation*}
We claim that $\Phi$ is satisfiable if and only if $\Gamma$ is. For
suppose $\fA \models \phi$. Define the truth-value assignment $\theta$ over the
proposition letters of $\Gamma$ by setting $\theta(o) = t$ just in
case $o_t^\fA = \emptyset$.  It follows from Lemma~\ref{lma:NP1} that,
if $o$ is any proposition letter mentioned in $\Gamma$, then
$\theta(o) = f$ just in case $o_f^\fA = \emptyset$. Now let $\gamma =
L_1 \vee L_2 \vee L_3$ be a clause in $\Gamma$. By
Lemma~\ref{lma:NP2}, for all $\gamma \in \Gamma$, there exists a $k$
($1 \leq k \leq 3$) such that $p^\fA_{\gamma_k} = \emptyset$.  By
Lemma~\ref{lma:NP3}: if $L_k = o$, then $o_t^\fA = \emptyset$, so that
$\theta(\gamma) = \theta(o) = t$; and if $L_k = \neg o$, then $o_f^\fA
= \emptyset$, so that $\theta(\gamma) = \theta(\neg o) = t$. Either
way, $\theta(\gamma) = t$.

Conversely, suppose $\theta$ is a truth-value assignment such that
$\theta(\gamma) = t$ for all $\gamma \in \Gamma$.  For all $o$
occurring in $\Gamma$, let $\fA_o$ be the structure
$\fA^{\theta(o)}_o$ over domain $A_o$ guaranteed by
Lemma~\ref{lma:NP1}. For each $\gamma = L_1 \vee L_2 \vee L_3 \in
\Gamma$, the set $K = \set{k \mid 1 \leq k \leq 3 \mbox{ and }
  \theta(L_k) = t}$ is non-empty; so let $\fA_\gamma$ be the structure
$\fA^K_\gamma$ over domain $A_\gamma$ guaranteed by
Lemma~\ref{lma:NP2}. Assume the domains of all these structures are
disjoint, and let
\begin{equation*}
\fB = \bigcup \set{\fA_o \mid o \mbox{ occurs in } \Gamma }
 \cup \bigcup \set{\fA_\gamma \mid \gamma \in \Gamma}.
\end{equation*}
Thus, for all $o$ occurring in $\Gamma$, $\fB$ agrees with $\fA_o$ on
the atoms $o_t$ and $o_f$, whence $\fB \models \Phi_o$. Likewise, for
all $\gamma \in \Gamma$, $\fB$ agrees with $\fA_\gamma$ on the atoms
in $\set{s_{\gamma,1}, s_{\gamma,2}, s_{\gamma,3}, s_{\gamma,4},
  p_{\gamma,1}, p_{\gamma,2}, p_{\gamma,3}}$, whence $\fB \models
\Phi_\gamma$. It remains to show that $\fB \models \Psi_\gamma$ for
each $\gamma \in \Gamma$. Suppose $\gamma = L_1 \vee L_2 \vee L_3 \in
\Gamma$, and $ 1 \leq k \leq 3$. If $L_k = o$, then $\psi_{\gamma,k} =
\forall(o_t,\overline{\forall p_{\gamma,k}})$. Take any $a \in
o_t^\fB$. By the construction of $\fB$, $a \in A_o$, and $\theta(o) =
f$, whence $\theta(L_k) = f$, whence, by the construction of $\fB$
again, $p_{\gamma,k}^\fB \cap A_\gamma \neq \emptyset$. Since
$A_o$ and $A_\gamma$ are disjoint, $\fB \models \psi_{\gamma,k}$.  On
the other hand, if $L_k = \neg o$, then $\psi_{\gamma,k} =
\forall(o_f,\overline{\forall p_{\gamma,k}})$. Take any $a \in
o_f^\fB$. By the construction of $\fB$, $a \in A_o$, and $\theta(o) =
t$, whence $\theta(L_k) = f$, whence, by the construction of $\fB$
again, $p_{\gamma,k}^\fB \cap A_\gamma \neq \emptyset$. Since
$A_o$ and $A_\gamma$ are disjoint, we again have $\fB \models
\psi_{\gamma,k}$.  Thus, $\fB \models \Phi$.  This establishes the
\NPTIME-hardness of the satisfiability problem for $\cH^{*\dagger}$.

To extend the result to $\cH^{\dagger}$, note that the only formulas
of $\Phi$ not in $\cH^{\dagger}$ are those the forms $\forall(\forall
o_t, \overline{\forall o_f})$ occurring in $\Phi_o$. But we can simply
replace any such formula, equisatisfiably, by the pair of formulas
$\forall(q, \overline{\forall o_f})$, $\forall(\bar{q},
\overline{\forall o_t})$, where $q$ is a fresh atom.
\end{proof}
\section{Complete indirect syllogistic systems for $\cH^\dagger$ and $\cH^{*\dagger}$}
\label{sec:complete}
The objective of this section is to prove
\begin{theorem}
There is a finite set ${\sH^\dagger}$ of syllogistic rules in
$\cH^\dagger$ such that the indirect derivation relation
$\Vdash_{\sH^\dagger}$ is sound and complete.
\label{theo:completeDagger}
\end{theorem}
\begin{theorem}
There is a finite set ${\sH^{*\dagger}}$ of syllogistic rules in
$\cH^{*\dagger}$ such that the indirect derivation relation
$\Vdash_{\sH^{*\dagger}}$ is sound and complete.
\label{theo:completeStarDagger}
\end{theorem}
We present first the proof of Theorem~\ref{theo:completeDagger}.  The
proof of Theorem~\ref{theo:completeStarDagger} proceeds similarly (and
in fact more simply); we indicate merely the differences between the
two proofs.

Let $\sH^{\dagger}$ consist of the following rules:
\begin{enumerate}
\item `little' rules:
\begin{equation*}
\begin{array}{ll}
\infer[(\mbox{\small I})]{\exists(\ell,\ell)}{\exists(\ell,c)}
\hspace{2cm} &
\infer[(\mbox{\small T})]{\forall(\ell,\ell)}{}\\ 
\ \\
\infer[(\mbox{\small A})]{\forall(c, m)}
                        {\forall(c,\ell) & \forall(c,\bar{\ell})}
&
\infer[(\mbox{\small N});]{\exists(\bar{\ell}, \bar{\ell})}
                        {\forall(\ell,\bar{\ell})}
\end{array}
\end{equation*}
\item generalizations of classical syllogisms:
\begin{equation*}
\begin{array}{lll}
\infer[(\mbox{\small B1})]{\forall(\ell, c)}
                        {\forall(\ell,m) & \forall(m,c)} & &
\infer[(\mbox{\small B2})]{\forall(\ell, m)}
                        {\forall(\ell,c) & \forall(c,m)} \\
\ \\
\infer[(\mbox{\small D1})]{\exists(\ell, c)}
                        {\exists(\ell,m) & \forall(m, c)}
& &
\infer[(\mbox{\small D2});]{\exists(\ell, m)}
                        {\exists(\ell,c) & \forall(c, m)}
\end{array}
\end{equation*}
\item the `Hamiltonian' rules:
\begin{equation*}
\begin{array}{lll}
\infer[{\mbox{\small (HH1)}}]{\forall (\ell, c)}
    {\exists (\ell, c) & \exists(m, \forall \ell)}
\hspace{1cm} 
&
\infer[{\mbox{\small (H2)}}]{\forall (\ell, \forall \ell)}
    {\exists (m, \forall \ell)}\\  
\ \\
\infer[{\mbox{\small (H3)}}]{\exists (m, \overline{\forall \ell})}
    {\exists (\ell, \overline{\forall m})}  
& 
\infer[{\mbox{\small (H4)}.}]{\forall(\bar{\ell}, \overline{\forall \ell})}
    {\exists (\ell, \ell)}
\hspace{0.5cm} 
\end{array}
\end{equation*}
\end{enumerate}
To avoid unnecessary proliferation of rule-names, those rules which
are simple generalizations of rules in $\sH$ have been given the same
names. Again, establishing the validity of the rules in
$\sH^{\dagger}$ is straightforward. Rule (A) is valid because its
premises imply that nothing is a $c$; we cannot replace (A) with the
simpler schema $\forall(c,\bar{c})/\forall(c,m)$, because, if $c$
is not a literal, $\forall(c,\bar{c})$ is not in the language
$\cH^\dagger$.  Rule (N)---no analogue of which can be formulated in
the language $\cH$---is valid because of the assumption that domains
are non-empty: if no $\ell$s are $\ell$s, then everything is a
non-$\ell$, and so something is a non-$\ell$. Rule (T) can in fact be
viewed as a special case of the rule (RAA), since we have the
derivation
\begin{equation*}
\infer[\mbox{\small{(RAA)}}^1.]{\forall(\ell,\ell)}{[\exists(\ell, \bar{\ell})]^1}
\end{equation*}
But we retain (T) as a separate rule for clarity.

Let $\Phi$ be a complete set of $\cH^\dagger$-formulas such that
$\Phi$ is consistent with respect to $\vdash_{\sH^\dagger}$.  In the
following lemmas, we build a structure $\fA$, and show that $\fA
\models \Phi$.  Since $\sH^{\dagger}$ is the only set of rules we
shall be concerned with in the ensuing lemmas, until further notice we
write $\vdash$ for the direct proof-relation $\vdash_{\sH^{\dagger}}$.

The elements of $A$ are constructed using sets of c-terms.  Call a set
$S$ of c-terms {\em consistent} if, for every c-term $c$, $c \in S$
implies $\bar{c} \not \in S$, and {\em literal-complete} if, for every
literal $\ell$, $\ell \not \in S$ implies $\bar{\ell} \in S$. Notice
that the notion of consistency for sets of c-terms is not the same as
$\vdash$-consistency for sets of formulas; likewise,
literal-completeness for sets of c-terms is not the same as
completeness for sets of formulas. Let $S$ be any set of
c-terms. Define
\begin{multline*}
S^* = S \cup \set{c \mid \mbox{there exists $\ell \in S$ such that $\Phi
                          \vdash \forall(\ell,c)$}} \cup \\
   \set{\ell \mid \mbox{there exists $c \in S$ such that $\Phi
                          \vdash \forall(c,\ell)$}}.
\end{multline*}
and we call $S$ {\em closed} if $S = S^*$. Trivially, $S \subseteq S^*$.
\begin{lemma}
Let $S$ be a set of c-terms. Then $S^*$ is closed.
\label{lma:dClosure}
\end{lemma}
\begin{proof}
We suppose $d \in (S^*)^* \setminus S^*$, and derive a contradiction.
We consider first the case where $d = m$ is a literal.  By definition,
there exists $c \in S^*$ such that $\Phi \vdash \forall(c,
m)$. Certainly, $c \not \in S$, for otherwise, we would have $d \in
S^*$.  Suppose first that $c$ is not a literal. Then there exists
$\ell \in S$ such that $\Phi \vdash \forall(\ell,c)$, and we have
the derivation
\begin{equation*}
\infer[\mbox{{\small (B2)}},]{\forall(\ell,m)}
  {\infer*{\forall(\ell, c)}{} & \infer*{\forall(c, m)}{}}
\end{equation*}
so that $m \in S^*$, a contradiction. On the other hand, suppose $c =
\ell$ is a literal. Then there exists a c-term $c_0 \in S$ such that
$\Phi \vdash \forall(c_0,\ell)$. Taking account of the equivalence of
$\forall(e,f)$ and $\forall(\bar{f},\bar{e})$, we have the derivation
\begin{equation*}
\infer[\mbox{{\small (B1)}},]{\forall(\bar{m}, \bar{c}_0)}
  {\infer*{\forall(\bar{m}, \bar{\ell})}{} & \infer*{\forall(\bar{\ell}, \bar{c}_0)}{}}
\end{equation*}
i.e.~$\Phi \vdash \forall(c_0,m)$, so that $m \in S^*$, a
contradiction.  The case where $d$ is not a literal proceeds similarly
(in fact, more simply).
\end{proof}
\begin{lemma}
Every closed, consistent set of c-terms containing at least one
literal has a closed, consistent, literal-complete extension.
\label{lma:dExtension}
\end{lemma}
\begin{proof}
Enumerate the literals as $\ell_0$, $\ell_1$, \ldots, and suppose $S$
is closed and consistent. Define $S^{(0)} = S$, and
\begin{equation*}
S^{(i+1)} = 
\begin{cases}
(S^{(i)} \cup \set{\ell_{i}})^* & \text{if $\bar{\ell}_i \not 
    \in S^{(i)}$}\\
S^{(i)} & \text{otherwise},
\end{cases}
\end{equation*}
for all $i \geq 0$.  It follows from Lemma~\ref{lma:dClosure} that
each $S^{(i)}$ is closed; we show by induction that it is also
consistent. From this it follows that $\bigcup_{0 \leq i} S^{(i)}$ is
consistent, thus proving the lemma. The case $i = 0$ is true by
hypothesis; so we suppose that $S^{(i)}$ is consistent, but
$S^{(i+1)}$ inconsistent, and derive a contradiction.  Let $m_0$ be a
literal in $S^{(0)}$, and hence in $S^{(i)}$; and let $c$ be a c-term
such that $c, \bar{c} \in S^{(i+1)}$.  Since $S^{(i)}$ is consistent,
by exchanging $c$ and $\bar{c}$ if necessary, we may assume that $c
\not \in S^{(i)}$. And since $S^{(i)}$ is also closed, we know that
either $c = \ell_i$ or $\Phi \vdash \forall(\ell_i, c)$. Indeed, by
rule (T), the latter case subsumes the former. Therefore, $\bar{c}
\not \in S^{(i)}$, since, otherwise, we would have $d = \bar{c} \in
S^{(i)}$ such that $\Phi \vdash \forall(d, \bar{\ell_i})$, whence
$\bar{\ell}_i \in S^{(i)}$, contrary to assumption. Since $\bar{c} \in
S^{(i+1)}$, it follows---again taking account of rule (T)---that $\Phi
\vdash \forall(\ell_i, \bar{c})$.  But then we have the derivation
\begin{equation*}
\infer[\mbox{{\small (A)}},]{\forall(\ell_i, \bar{m}_0)}
 {\infer[\mbox{{\small (T)}}]{\forall(\ell_i, \ell_i)}{}
  &
  \infer[\mbox{{\small (B2)}}]{\forall(\ell_i, \bar{\ell}_i)}
    {\infer*{\forall(\ell_i, c)}{} & \infer*{\forall(c, \bar{\ell}_i)}{}}}
\end{equation*}
so that $\Phi \vdash \forall(m_0, \bar{\ell}_i)$, again contrary to
the fact that $S^{(i)} \neq S^{(i+1)}$.
\end{proof}

Denote by $W$ the set of all closed, consistent and literal-complete
sets of c-terms.  In the sequel, we use the variables $u$, $v$, $w$ to
range over $W$.
\begin{lemma}
Suppose $\Phi \vdash \exists(\ell,c)$. Then there exists $w \in W$ such
that $\ell, c \in w$.
\label{lma:dL1}
\end{lemma}
\begin{proof}
By Lemmas~\ref{lma:dClosure} and~\ref{lma:dExtension}, we need only
show that $\set{\ell,c}^*$ is consistent. So suppose otherwise.  Since
$\Phi$ is $\vdash$-consistent, $c \neq \bar{\ell}$. We
therefore have the following possible cases: ({\em i}) $\Phi \vdash
\forall(\ell,\bar{c})$; ({\em ii}) there exists $d$ such that $\Phi
\vdash \forall(\ell,d)$ and $\Phi \vdash \forall(\ell,\bar{d})$; ({\em
  iii}) there exists $d$ such that $\Phi \vdash \forall(\ell,d)$ and
$\Phi \vdash \forall(c,\bar{d})$; ({\em iv}) there exists $d$ such
that $\Phi \vdash \forall(c,d)$ and $\Phi \vdash
\forall(c,\bar{d})$. Note that, in Cases ({\em iii}) and ({\em iv}),
one of $c$ or $d$ must be a literal.  In Case ({\em i}), Rule (D2)
immediately yields $\Phi \vdash \exists(\ell,\bar{\ell})$.  In Case
({\em ii}), we have the derivation:
\begin{equation*}
\infer[{\mbox{\small (D2)}.}]{\exists (\ell, \bar{\ell})}
    {\infer[{\mbox{\small (D1)}}]{\exists (\ell,\bar{d})}
                 {\infer[{\mbox{\small (I)}}]{\exists(\ell,\ell)}
                                               {\infer*{\exists(\ell,c)}{}} & 
                  \infer*{\forall(\ell,\bar{d})}{}} 
           & \infer*{\forall (\bar{d},\bar{\ell})}{}}  
\end{equation*}
Likewise, in case ({\em iii}), we have the derivation:
\begin{equation*}
\infer[{\mbox{\small (D1) or (D2)}.}]{\exists (\ell, \bar{\ell})}
    {\infer[{\mbox{\small (D1) or (D2)}}]{\exists (\ell,\bar{d})}
                 {\infer*{\exists(\ell,c)}{} & 
                  \infer*{\forall(c,\bar{d})}{}} 
           & \infer*{\forall (\bar{d},\bar{\ell})}{}}  
\end{equation*}
In Case ({\em iv}), if $c$ is a literal, we proceed as in Case ({\em
  ii}), but with $\ell$ and $c$ exchanged; and if $d = m$ is a literal, we
have the derivation:
\begin{equation*}
\infer[{\mbox{\small (D2)}.}]{\exists (\ell, \bar{\ell})}
   {\infer*{\exists (\ell, c)}{}
    &
    \infer[{\mbox{\small (A)}}]{\forall (c,\bar{\ell})}
        {\infer*{\forall (c,m)}{}
         &
         \infer*{\forall (c,\bar{m})}{}
         }
    }
\end{equation*}
Since all cases contradict the supposed
$\vdash$-consistency of $\Phi$, the lemma is proved.
\end{proof}
\begin{lemma}
The set $W$ is not empty. 
\label{lma:dNonEmpty}
\end{lemma}
\begin{proof}
By Lemma~\ref{lma:dL1}, it is necessary only to show that $\Phi \vdash
\exists(\ell,c)$ for some $\ell$ and $c$. Pick any $\ell$. If
$\exists(\ell, \ell) \in \Phi$, we are done. Otherwise, by
completeness of $\Phi$, $\forall(\ell, \bar{\ell}) \in \Phi$, so that,
by Rule (N), $\Phi \vdash \exists(\bar{\ell}, \bar{\ell})$, completing
the proof.
\end{proof}
The following lemma is the analogue, for the system $\sH^\dagger$, of
Lemma~\ref{lma:5}. This time, however, the lemma is trivial, because
we are assuming that $\Phi$ is complete.
\begin{lemma}
Suppose $\ell, c \in w \in W$. Then $\exists(\ell,c) \in \Phi$.
\label{lma:dL2}
\end{lemma}
\begin{proof}
Suppose $\exists(\ell,c) \not \in \Phi$.  By the completeness of
$\Phi$, $\forall(\ell,\bar{c}) \in \Phi$, whence $\bar{c} \in w$,
because $w$ is closed. This contradicts the consistency of $w$.
\end{proof}
\begin{lemma}
Suppose $w \in W$, and $(\forall \ell) \in w$, where $\Phi \vdash
\exists(\ell,\ell)$. Then $\ell \in w$.
\label{lma:dL2.5}
\end{lemma}
\begin{proof}
Suppose otherwise. By the literal-completeness of $w$, $\bar{\ell} \in
w$.  But we have the derivation
\begin{equation*}
\infer[{\mbox{\small (H4)},}]{\forall(\bar{\ell}, \overline{\forall \ell})}
    {\exists (\ell, \ell)}
\end{equation*}
so that, since $w$ is closed, $\overline{\forall \ell} \in w$,
contradicting the consistency of $w$.
\end{proof}
\begin{lemma}
Suppose $u, v, w \in W$ with $(\forall \ell) \in u$, $(\forall \ell) \in v$
and $\ell \in w$. Then $u = v$.
\label{lma:dL3}
\end{lemma}
\begin{proof}
By Lemma~\ref{lma:dL2}, $\Phi \models \exists(\ell, \ell)$. By
Lemma~\ref{lma:dL2.5}, $\ell \in w$ and $\ell \in v$. Suppose also $c
\in u$.  By Lemma~\ref{lma:dL2} again, $\exists(\ell, c) \in \Phi$,
and $\exists(\ell, \forall \ell) \in \Phi$. Therefore, we have the
derivation
\begin{equation*}
\infer[{\mbox{\small (HH1)},}]{\forall(\ell, c)}
      {\exists(\ell, c) &
       \exists (\ell, \forall \ell)} 
\end{equation*}
and $c \in v$. Thus, $u \subseteq v$. The reverse inclusion follows
symmetrically.
\end{proof}

Analogously to Section~\ref{sec:refutationComplete}, we call $w \in W$
{\em special} if it contains a c-term of the form $\forall \ell$ such
that $\Phi \vdash \exists(\ell, \ell)$; and we build the structure
$\fA$ as follows:
\begin{eqnarray*}
A & = &\set{\langle w, 0\rangle \mid w \in W \mbox{ is special}} \cup\\
& & \qquad \set{\langle w, i \rangle \mid w \in W \mbox{ is non-special, } i \in \set{-1,1}}\\ 
p^\fA & = &\set{ \langle w,i \rangle \in A \mid p \in w}, \text{ for any atom $p$.}
\end{eqnarray*}
We remark that, since, by Lemma~\ref{lma:dNonEmpty}, $W$ is non-empty,
$A$ is non-empty; so this construction is legitimate.
\begin{lemma}
Suppose $c$ is a c-term and $a = \langle w,i \rangle$. Then $c \in w$
implies $a \in c^\fA$.  Further, if $\ell$ is a literal, Then $a \in
\ell^\fA$ implies $\ell \in w$.
\label{lma:dL4}
\end{lemma}
\begin{proof}
We consider the possible forms of $c$ in turn.

\bigskip

\noindent
1. $c = p$ is an atom: By construction of $\fA$, $c \in w$ if and only
if $p \in w$.

\bigskip

\noindent
2. $c = \overline{p}$: By consistency and literal-completeness of $w$,
$c \in w$ if and only if $p \not \in w$. The result then follows by
Case 1.

\bigskip

\noindent
3. $c = \forall \ell$: Suppose $c \in w$, and $a' = \langle w', i'
\rangle$ is such that $a' \in \ell^\fA$. By Cases~1 and~2, $\ell \in
w'$. Pick any literal $m \in w$. By Lemma~\ref{lma:dL2},
$\exists(\ell, \ell) \in \Phi$ and $\exists(m, \forall \ell) \in
\Phi$.  Thus, we have the derivation
\begin{equation*}
\infer[{\mbox{\small (H2)},}]{\forall (\ell, \forall \ell)}
    {\exists (m, \forall \ell)}  
\end{equation*}
whence $\forall \ell \in w'$, and therefore, by Lemma~\ref{lma:dL3},
$w = w'$.  Indeed, since $w$ is special, the construction of $A$
ensures that $i = i' = 0$, and hence $a = a'$. Thus, $a' \in \ell^\fA$
implies $a = a'$, whence $a \in c^\fA$, as required.

\bigskip

\noindent
4. $c = \overline{\forall \ell}$: Suppose $c \in w$, and assume for
the time being that $i \neq 0$.  Pick any literal $m \in w$. By
Lemma~\ref{lma:dL2}, $\exists(m,\overline{\forall \ell}) \in \Phi$, so
that we have the derivation
\begin{equation*}
\infer[{\mbox{\small (H3)}.}]{\exists (\ell, \overline{\forall m})}
    {\exists (m, \overline{\forall \ell})}
\end{equation*}
By Lemma~\ref{lma:dL1}, there exists $w' \in W$ such that $\ell \in
w'$, and by construction of $A$ and Cases~1 and~2 above, there exists
$i' \in \set{-1,0,1}$ such that both $\langle w', i' \rangle \in
\ell^\fA$ and $\langle w', -i' \rangle \in \ell^\fA$. Since $i \neq 0$
we may suppose $i \neq i'$, so that there exists $a' \in A$ with $a'
\neq a$ and $a' \in \ell^\fA$. Hence $a \in c^\fA$, as required.  Now
assume $i = 0$. Then $w$ is special, so suppose $(\forall m) \in w$,
with $\Phi \vdash \exists(m,m)$. By Lemma~\ref{lma:dL2.5}, $m \in w$,
so that, by Lemma~\ref{lma:dL2}, $\exists(m,\overline{\forall \ell})
\in \Phi$.  Again, then, by (H3) and Lemma~\ref{lma:dL1}, there
exists $w' \in W$ such that $\ell \in w'$ and also $\overline{\forall
  m} \in w'$.  By construction of $\fA$ and Cases~1 and~2 above, there
exists $i' \in \set{-1,0,1}$ such that $\langle w', i' \rangle \in
\ell^\fA$.  Since $\forall m \in w$, we have $w \neq w'$, and
therefore $a \neq a'$. Hence $a \in c^\fA$, as required.
\end{proof}
\begin{proof}[Proof of Theorem~\ref{theo:completeStarDagger}]
Since we are dealing with an indirect proof relation, it suffices to
show that every $\Vdash_{\sH^\dagger}$-consistent set of formulas is
true in some structure. Let $\Phi$ be
$\Vdash_{\sH^\dagger}$-consistent. By Lemma~\ref{lma:lindenbaum}, we
may further assume without loss of generality that $\Phi$ is
complete. Certainly, $\Phi$ is $\vdash_{\sH^\dagger}$-consistent.  Let
$\fA$ be constructed as described above: we show that $\fA \models
\Phi$.  For suppose $\phi = \exists(\ell,c) \in \Phi$. By
Lemma~\ref{lma:dL1}, there exists $a = \langle w,i \rangle \in A$ such
that $\ell, c \in w$. By Lemma~\ref{lma:dL4}, $a \in \ell^\fA$ and $a
\in c^\fA$; thus, $\fA \models \phi$.  On the other hand, suppose
$\phi = \forall(\ell,c) \in \Phi$. If $a = \langle w,i \rangle \in
\ell^\fA$, then, by (the second statement of) Lemma~\ref{lma:dL4},
$\ell \in w$, whence, by the fact that $w$ is closed, $c \in w$,
whence $a \in c^\fA$, by Lemma~\ref{lma:dL4}; thus, $\fA \models
\phi$.
\end{proof}

Turning now to the language $\cH^{*\dagger}$, let $\sH^{*\dagger}$
consist of the following rules:
\begin{enumerate}
\item `little' rules:
\begin{equation*}
\infer[(\mbox{\small I})]{\exists(e,e)}{\exists(e,f)}
\hspace{1cm}
\infer[(\mbox{\small T})]{\forall(e,e)}{}
\hspace{1cm}
\infer[(\mbox{\small A})]{\forall(f, \bar{e})}
                        {\forall(e,\bar{e})}
\hspace{1cm}
\infer[(\mbox{\small N});]{\exists(\bar{e}, \bar{e})}
                        {\forall(e,\bar{e})}
\end{equation*}
\item generalizations of classical syllogisms:
\begin{equation*}
\infer[(\mbox{\small B})]{\forall(e, g)}
                        {\forall(e,f) & \forall(f,g)}
\hspace{1cm}
\infer[(\mbox{\small D});]{\exists(e, g)}
                        {\exists(e,f) & \forall(f, g)}
\end{equation*}
\item the `Hamiltonian' rules:
\begin{equation*}
\begin{array}{lll}
\infer[{\mbox{\small (HH1)}}]{\forall (\ell, e)}
    {\exists (\ell, e) & \exists(m, \forall \ell)}
\hspace{1cm} 
&
\infer[{\mbox{\small (H2)}}]{\forall (\ell, \forall \ell)}
    {\exists (e, \forall \ell)}\\  
\ \\
\infer[{\mbox{\small (H3)}}]{\exists (m, \overline{\forall \ell})}
    {\exists (\ell, \overline{\forall m})}  
& 
\infer[{\mbox{\small (H4)}.}]{\forall(\bar{\ell}, \overline{\forall \ell})}
    {\exists (\ell, \ell)}
\hspace{0.5cm} 
\end{array}
\end{equation*}
\end{enumerate}
Where rules in $\sH^{*\dagger}$ are obvious generalizations of
counterparts in $\sH^{\dagger}$, we have kept the same names.
Otherwise, $\sH^{*\dagger}$ is simpler than $\sH^{\dagger}$: in
particular, Rule (A) now has only one premise, and Rules (B1) and (B2)
have been subsumed under the more general Rule (B); similarly for (D1)
and (D2).  The proof that $\Vdash_{\sH^{*\dagger}}$ is complete for
$\cH^{*\dagger}$ proceeds as for Theorem~\ref{theo:completeDagger},
the essential difference being that various complications arising from
the restricted syntax of $\sH^{\dagger}$ disappear. Consequently, we
confine ourselves to a proof sketch.

Let $\Phi$ be a complete set of $\cH^{*\dagger}$-formulas such that
$\Phi$ is $\vdash_{\sH^{*\dagger}}$-consistent.  We build a structure
$\fA$, and show that $\fA \models \Phi$, this time writing $\vdash$ to
mean $\vdash_{\sH^{*\dagger}}$.  The elements of $A$ are constructed
using sets of e-terms.  Call a set of e-terms $S$ {\em consistent} if,
for every e-term $e$, $e \in S$ implies $\bar{e} \not \in S$, and {\em
  term-complete} if, for every e-term $e$, $e \not \in S$ implies
$\bar{e} \in S$. If $S$ is a set of e-terms, define
\begin{equation*}
S^* = \set{f \mid \mbox{there exists $e \in S$ such that $\Phi \vdash \forall(e,f)$}},
\end{equation*}
and we call $S$ {\em closed} if $S = S^*$.  Note that the definition
of $S^*$ for the system $\sH^{*\dagger}$ is simpler than the
corresponding definition for $\sH^{\dagger}$.  For any set $S$ of
e-terms, it is immediate from Rule (T) that $S \subseteq S^*$, and
immediate from Rule (B) that $S^*$ is closed (the analogue of
Lemma~\ref{lma:dClosure}).  We now define $W$ to be the set of all
closed, consistent and term-complete sets of e-terms; and we show,
analogously to Lemmas~\ref{lma:dL1}--\ref{lma:dL2}, that $W$ is
non-empty, and that, for any e-terms $e$ and $f$, $\exists(e,f) \in
\Phi$ if and only if there exists $w \in W$ such that $e, f \in
w$. Further, by (HH1) and (H4), we easily show, analogously to
Lemma~\ref{lma:dL3}, that if $u, v, w \in W$ with $(\forall \ell) \in
u$, $(\forall \ell) \in v$ and $\ell \in w$, then $u = v$. Defining
\begin{eqnarray*}
A & = &\set{\langle w, 0\rangle \mid w \in W \mbox{ is special}} \cup\\
& & \qquad \set{\langle w, i \rangle \mid w \in W \mbox{ is non-special, } i \in \set{-1,1}}\\ 
p^\fA & = &\set{ \langle w,i \rangle \in A \mid p \in w}, \text{ for any atom $p$,}
\end{eqnarray*}
we show, analogously to Lemma~\ref{lma:dL4}, that, for any e-term $e$
and any domain element $a = \langle w,i \rangle$, $a \in e^\fA$ if and
only if $e \in w$. Note that this is a stronger statement than
Lemma~\ref{lma:dL4}, and uses the fact that $w$ is term-complete, not
just literal-complete. The remainder of the argument then proceeds as
for Theorem~\ref{theo:completeDagger}, but exploiting the
fact that, if $e$ is an arbitrary e-term (not just a literal) and
$\langle a, i \rangle \in e^\fA$, then $e \in w$.

We finish with a proof of the claim made, in passing, at the end of
Section~\ref{sec:synsem}, regarding models of sets of
$\cH^{*\dagger}$-formulas.
\begin{theorem}
Let $\Phi$ be a set of $\cH^{*\dagger}$-formulas. If $\Phi$ has a model
with three or more elements, then it has arbitrarily large models.
\label{theo:largeModels}
\end{theorem}
\begin{proof}
Again, we may assume without loss of generality that $\Phi$ is a
complete set of formulas.  If $S$ is a set of $e$-terms, we use the
notation $S^*$ in the sense of the above sketch proof of
Theorem~\ref{theo:completeStarDagger}.  Suppose $\fB \models \Phi$,
with $\omega > |B| \geq 3$.  Write $B = \set{b_1, \ldots, b_n}$. We
may assume that each $b_i$ is the unique element satisfying some
literal $\ell_i$, since, otherwise, we can add as many duplicate
copies of $b_i$ to $\fB$ as we like without affecting the truth of any
$\cH^{*\dagger}$-formulas. It follows that, for all $i$ ($1 \leq i
\leq n$), the set of e-terms $\set{\bar{\ell}_1, \ldots,
  \bar{\ell}_{i-1}, \ell_i, \bar{\ell}_{i+1}, \ldots,
  \bar{\ell}_{n}}^*$ is consistent.  We claim that $\set{\bar{\ell}_1,
  \ldots, \bar{\ell}_n}^*$ is also consistent. For otherwise, it is
immediate from rule (B) that, for some $j$, $k$ ($1 \leq j \leq k \leq
n$), $\Phi \vdash_{\cH^{*\dagger}} \forall (\bar{\ell}_j, \ell_k)$,
contradicting the assumption that $b_k$ is the unique element of $\fB$
satisfying $\ell_k$ (remember that $n \geq 3$). Now let $\fA$ be the
model constructed in the proof of
Theorem~\ref{theo:completeStarDagger}. Since $\set{\bar{\ell}_1,
  \ldots, \bar{\ell}_{n}}^*$ is consistent, it has a consistent
complete extension, $w$, so that $\fA$ contains some element $a =
\langle w, h \rangle$ satisfying $\bar{\ell}_1$, \ldots,
$\bar{\ell}_n$.  But, by the same token, $\fA$ also contains an
element $a_i$ satisfying $\ell_j$ if and only if $i = j$. Thus, $\fA$
has cardinality at least $n+1$.
\end{proof}
\section*{Acknowledgements}
The author wishes to express his gratitude to Lawrence S.~Moss for
comments on an earlier version of this paper.
\bibliographystyle{plain} \bibliography{hs}
\end{document}